\documentclass[letterpaper,twocolumn,10pt]{article}
\usepackage{preprint}

\usepackage{cite}
\usepackage{amsmath}
\usepackage{algorithmic}
\usepackage{url}
\usepackage[ruled, vlined, linesnumbered]{algorithm2e}

\usepackage{amsthm}
\usepackage{amssymb}

\usepackage{subcaption}

\usepackage{booktabs}
\usepackage{graphicx}

\usepackage{enumitem}

\usepackage{tcolorbox}

\usepackage{pifont}

\usepackage{makecell}

\theoremstyle{definition}
\newtheorem{definition}{Definition}[section]

\newtheorem{theorem}{Theorem}[section]
\newtheorem{lemma}{Lemma}[section]

\makeatletter
\newcommand{\Upon}[1]{\STATE \textbf{upon} #1 \textbf{do}\begin{ALC@g}}
\newcommand{\EndUpon}{\end{ALC@g}\STATE \textbf{end upon}}
\makeatother

\SetKwFor{AlgUpon}{upon}{do}{}
\SetKwFor{AlgConstruct}{construct}{with}{}
\makeatletter
\newcommand{\AlgTwoEUnnumbered}[1]{\par
  \begingroup
  \algocf@seteveryparnl{\relax}#1\endgroup
}
\newcommand{\AlgTwoERole}[1]{\AlgTwoEUnnumbered{\(\rhd\) #1\;}}
\newlength{\AlgTwoESepVSkip}
\newcommand{\AlgTwoESep}{\par
  \vspace{\dimexpr\AlgTwoESepVSkip-\skiphlne\relax}\nointerlineskip
  \hbox to \hsize{\hskip-\algomargin
    \rule{\algowidth}{0.4pt}\hss
  }\nointerlineskip
  \vspace{\AlgTwoESepVSkip}}
\makeatother

\newcommand{\QuarkShard}{\textsc{QuarkShard}}

\newcommand{\setS}{\mathcal{S}}

\newcommand{\Chain}[1]{C_{#1}}
\newcommand{\Shard}[1]{S_{#1}}
\newcommand{\Block}[2]{B_{#1}^{#2}}

\newcommand{\bal}[2]{\mathsf{bal}_{#1}(#2)}

\newcommand{\DeltaSet}[3]{\mathsf{CB}_{#1\to #2}^{#3}}

\newcommand{\DeltaProof}[3]{\pi_{#1\to #2}^{#3}}

\newcommand{\send}{\mathsf{Send}}
\newcommand{\deliver}{\mathsf{Deliver}}
\newcommand{\verify}{\mathsf{Verify}}
\newcommand{\csrd}[2]{\mathsf{CSRD}_{#1\to #2}}
\newcommand{\Send}[1]{\mathsf{Send}(#1)}
\newcommand{\Deliver}[1]{\mathsf{Deliver}(#1)}
\newcommand{\Verify}[1]{\mathsf{Verify}(#1)}

\newcommand{\ECEnc}{\mathsf{EC.Encode}}
\newcommand{\ECDec}{\mathsf{EC.Decode}}
\newcommand{\COM}{\mathsf{Com}}
\newcommand{\COMCommitOpen}{\COM.\mathsf{CommitOpen}}
\newcommand{\COMVerify}{\COM.\mathsf{Verify}}
\newcommand{\TSig}{\mathsf{TSig}}
\newcommand{\TSigSignShare}[2]{\TSig.\mathsf{SignShare}_{#1}(#2)}
\newcommand{\TSigVerifyShare}[3]{\TSig.\mathsf{VerifyShare}(#1,#2,#3)}
\newcommand{\TSigCombine}[2]{\TSig.\mathsf{Combine}(#1,#2)}
\newcommand{\TSigVerify}[2]{\TSig.\mathsf{Verify}(#1,#2)}

\newcommand{\WaitUntil}{\textbf{wait until}}

\newcommand{\BFTLatency}{L_{\mathrm{BFT}}}
\newcommand{\CSRDLatency}{L_{\mathrm{CSRD}}}
\newcommand{\ImportLatency}{L_{\mathrm{imp}}}

\newcommand{\anglebk}[1]{\left\langle #1 \right\rangle}
\newcommand{\bigbk}[1]{\left\{ #1 \right\}}

\begin{document}

\date{}

\title{\Large \bf \QuarkShard{}: Decomposing Cross-Shard Requests into Single-Input Transactions}

\author{
{\rm Changyu Duan}\\
Tsinghua University
\and
{\rm Zhixuan Fang}\\
Tsinghua University
}
\maketitle

\begin{abstract}
Sharding blockchains enhance scalability by distributing the workload across multiple shards.
Within sharding blockchains, securely and efficiently processing cross-shard user requests remains a pivotal challenge.
Traditional sharding blockchain protocols execute cross-shard requests under a processing paradigm that treats each user request as a monolithic entity to preserve its atomicity, inherently introducing significant intra-shard BFT execution overhead and cross-shard communication overhead.

We propose \QuarkShard{}, a sharding protocol that breaks away from this traditional paradigm to achieve intra-shard and cross-shard overheads below the lower bounds imposed by conventional approaches.
\QuarkShard{} decomposes cross-shard requests into single-input-shard transactions, where each transaction requires consensus and processing within only a single shard, thereby substantially reducing total request processing overhead.
Meanwhile, through a carefully designed request decomposition algorithm, \QuarkShard{} guarantees fund atomicity or location atomicity during request execution.
\QuarkShard{} also enhances the security of blockchain sharding systems against MEV attacks and flooding attacks.
Empirical evaluation across a wide-area network spanning four AWS regions demonstrates that \QuarkShard{} achieves superior performance, attaining a peak throughput of $1024$\,k requests per second with sub-$4$-second latency.
Benchmark comparisons show that \QuarkShard{} achieves up to $210.7\%$ and $314.8\%$ of the maximum throughputs of Kronos and 2PC, respectively.
\end{abstract}

\section{Introduction}

As blockchain technology continues to evolve, enhancing scalability has become an important challenge.
Among various solutions, sharding architectures are widely regarded as an effective approach~\cite{kokoris2018omniledger,zamani2018rapidchain,al2017chainspace,wang2019monoxide,huang2022brokerchain,hellings2021byshard,liu2024cherubim,liu2024kronos,liu2026logos,dang2019towards} and have been implemented in multiple industrial-grade systems~\cite{harmony2018,skidanov2019nightshade,polkadot2020}.

In traditional blockchains, every replica must store the complete blockchain state and participate in a global consensus protocol to process all transactions across the network.
In contrast, sharding partitions both the blockchain state and the transaction execution process across multiple shards.
Each shard is maintained by a subset of replicas and processes only a fraction of the system workload, thereby significantly reducing the per-replica load.
By enabling horizontal scaling, sharding allows throughput to scale dynamically as more nodes and shards join the system.

A central challenge in sharding blockchains is how to process user requests that access state across multiple shards without compromising security or scalability.
In sharding blockchains, each shard maintains only part of the global state, such as a subset of account balances.
Therefore, a user may intend to draw funds from multiple shards and pay accounts located in multiple destination shards.
Prior studies have shown that such cross-shard payments can dominate the workload: when the system has 16 shards, more than \(99\%\) of requests would access multiple shards~\cite{zamani2018rapidchain}.
For clarity, we use \emph{request} for the user's payment intent and \emph{transaction} for the object processed by the sharding blockchain protocol in this work.\footnote{Prior work often uses transaction-related terminology for both notions; we separate them because \QuarkShard{} deliberately decouples user requests from the transactions processed by the sharding blockchain protocol.}

To process a cross-shard user request, existing works typically require coordinating the involved shards to reach consensus on and execute the request according to specific rules.
Ultimately, this coordination must guarantee request atomicity: the request should either take effect on all involved shards, or leave the system as if the request had not been executed.
Specifically, the execution of a cross-shard request can be decomposed into two main components: intra-shard consensus and cross-shard communication.

Generally, sharding blockchains run a Byzantine Fault Tolerance (BFT) protocol within each shard to process and reach consensus on transactions.
Processing each transaction incurs corresponding communication, computation, and storage overheads, which we collectively refer to as the intra-shard consensus overhead (IS-overhead).
Many current protocols rely on two-phase commit (2PC) based schemes~\cite{al2017chainspace,dang2019towards,hellings2021byshard,liu2024cherubim}, incurring an IS-overhead of \((2a+b)\mathcal{B}\) or larger for a request involving \(a\) input shards and \(b\) output shards, where \(\mathcal{B}\) denotes the intra-shard BFT execution cost required to process a single transaction.
State-of-the-art general cross-shard protocols, such as Kronos~\cite{liu2024kronos} and Logos~\cite{liu2026logos}, reduce this overhead to \((a+b)\mathcal{B}\).

To synchronize execution results across multiple shards, shards must transmit cryptographic certificates to one another, incurring what we term cross-shard communication overhead (CS-overhead).
Existing protocols~\cite{dang2019towards,hellings2021byshard,al2017chainspace} typically require all-to-all inter-shard broadcasting to guarantee message reachability, introducing a CS-overhead of \(O(n^2r\lambda)\), where \(n\) denotes the number of replicas per shard, \(r\) represents the number of transactions in a batch, and \(\lambda\) is cryptographic parameter size.
Kronos~\cite{liu2024kronos} reduces this overhead to \(O(nr\lambda)\) by leveraging erasure codes~\cite{luby2001efficient}.
Logos~\cite{liu2026logos} further lowers the pure cross-shard communication overhead to \(O(r\lambda)\) by offloading most of the communication into intra-shard messaging (though the total communication overhead remains \(O(nr\lambda)\)).

These traditional cross-shard protocols share a common processing paradigm: the protocol treats each user request as a monolithic entity for consensus and execution, directly guaranteeing its atomicity through protocol-level mechanisms.
However, this paradigm imposes inherent overheads. Prior work has proven that the intra-shard consensus overhead cannot fall below \((a+b)\mathcal{B}\)~\cite{liu2024kronos,liu2026logos}, as each involved shard must execute at least one consensus round for the request.
Similarly, because a certificate must be broadcast to recipient shards for every request, the lower bound for the total communication overhead remains \(\Omega(nr\lambda)\), and the lower bound for the CS-overhead is \(\Omega(r\lambda)\)~\cite{liu2026logos}.

These lower bounds motivate us to reconsider a fundamental question: \textbf{Is treating each user request as an indivisible entity strictly necessary?}

\begin{table*}[t]
\centering
\caption{Comparison of \QuarkShard{} and BFT-based Sharding Blockchain Systems.}
\label{tab:protocol-comparison}
\footnotesize
\setlength{\tabcolsep}{13.3pt}
\renewcommand{\arraystretch}{1.12}
\newcommand{\ProtocolComparisonHeader}[1]{\raisebox{-1.4ex}[0pt][0pt]{\(\vcenter{\hbox{\shortstack[c]{#1}}}\)}}
\begin{tabular}{c|c|c|c|c|c}
\hline
\rule{0pt}{3ex}\ProtocolComparisonHeader{\textbf{Protocol}}
& \multicolumn{2}{c|}{\textbf{IS-Overhead}\textsuperscript{\(\dagger\)}}
& \ProtocolComparisonHeader{\textbf{CS-Overhead}\textsuperscript{\(\ddagger\)}\\(for \(r\) transactions)}
& \ProtocolComparisonHeader{\textbf{CS Front-Running}\\\textbf{Free}\textsuperscript{\(\circ\)}}
& \ProtocolComparisonHeader{\textbf{Order-Fairness BFT}\\\textbf{Compatible}\textsuperscript{\(\star\)}}
\\
\cline{2-3}
\rule{0pt}{2.5ex}
& \textbf{valid Request}
& \textbf{invalid Request}
&
&
&
\\
\hline
OmniLedger~\cite{kokoris2018omniledger}
& \((a+b)\mathcal{B}\)
& \(2a\mathcal{B}\)
& \(O(n^2r\lambda)\)
& \ding{56}
& \ding{56}
\\
RapidChain~\cite{zamani2018rapidchain}
& \((a+b)\mathcal{B}\)\textsuperscript{\(\diamond\)}
& \(x\mathcal{B}\)\textsuperscript{\(\diamond\)}
& \(O(n^2r\lambda)\)
& \ding{56}
& \ding{56}
\\
Chainspace~\cite{al2017chainspace}
& \((2a+b)\mathcal{B}\)
& \(2a\mathcal{B}\)
& \(O(n^2r\lambda)\)
& \ding{56}
& \ding{56}
\\
ByShard~\cite{hellings2021byshard}
& \((2(a+b)-1)\mathcal{B}\)
& \((2x+1)\mathcal{B}\)
& \(O(n^2r\lambda)\)
& \ding{56}
& \ding{56}
\\
AHL~\cite{dang2019towards}
& \((2(a+b)+3)\mathcal{B}\)
& \((2(a+b)+3)\mathcal{B}\)
& \(O(n^2r\lambda)\)
& \ding{56}
& \ding{56}
\\
Kronos~\cite{liu2024kronos}
& \((a+b)\mathcal{B}\)
& \(2x\mathcal{B}\)
& \(O(nr\lambda)\)
& \ding{56}
& \ding{56}
\\
Logos~\cite{liu2026logos}
& \((a+b)\mathcal{B}\)
& \(a\mathcal{C}\)
& \(O(r\lambda)\)
& \ding{56}
& \ding{56}
\\
\hline
\QuarkShard{}
& \(a\mathcal{B}\)
& \makecell[c]{\(x\mathcal{B}\)\\\((x+\mathbf{1}[x>0])\mathcal{B}\)\textsuperscript{\(\S\)}}
& \makecell[c]{\(O(nc\lambda)\)\\\(O(c\lambda)\)\textsuperscript{\(\S\)}}
& \ding{52}
& \ding{52}
\\
\hline
\end{tabular}
\begin{minipage}{\textwidth}
\scriptsize
\vspace{0.6em}
\ding{52} and \ding{56} denote that a property is satisfied and not satisfied, respectively.
\begin{description}[leftmargin=1.8em,labelwidth=1.25em,labelsep=0.35em,itemsep=0.05em,topsep=0.15em]
    \item[\(\dagger\)] \(a\) and \(b\): the numbers of input and output shards for a request. \(\mathcal{B}\): the BFT execution cost for a single transaction. \(x\) \((< a)\): the number of input shards that have consumed or locked funds prior to the rejection of an invalid request.
    \item[\(\ddagger\)] \(r\): the number of transactions in a batch. \(c\): the number of output accounts involved in these transactions. \(\lambda\): the cryptographic parameter size. This column indicates the communication overhead required to guarantee that every honest node in the destination shard eventually receives the identical, reliably delivered data~\cite{liu2026logos}.
    \item[\(\circ\)] \emph{CS Front-Running Free} indicates whether the cross-shard request processing scheme eliminates time windows in which an adversary can observe transaction intent in advance and perform front-running using intra-shard transactions (Definition~\ref{def:CSFFF}).
    \item[\(\star\)] \emph{Order-Fairness BFT Compatible} indicates whether the system is compatible with non-sharding Order-Fairness BFT consensus protocols (Definition~\ref{def:OFBC}).
    \item[\(\diamond\)] RapidChain only supports single-destination output requests (\(b=1\)) and does not support all-or-nothing atomicity. The term \(x\mathcal{B}\) represents its overhead for invalid requests when achieving fund atomicity analogous to that defined in this paper.
    \item[\(\S\)] The terms \(x\mathcal{B}\) and \((x+\mathbf{1}[x>0])\mathcal{B}\) represent the IS-overhead under fund atomicity and location atomicity, respectively. \(\mathbf{1}[x>0]\) is the indicator function that equals 1 if \(x>0\) and 0 otherwise. The term \(O(c\lambda)\) represents the reduced overhead when employing a Logos-based optimization scheme~\cite{liu2026logos}.
\end{description}
\end{minipage}
\end{table*}

The answer is negative. We propose \QuarkShard{}, which splits each cross-shard user request into finer-grained pieces, where each piece is a transaction with a unique input shard, termed a \emph{Quark transaction}.
The protocol no longer reaches consensus on monolithic user requests; instead, it processes only Quark transactions, while request-level atomicity is guaranteed via a carefully designed request splitting algorithm.
Through this decomposition, \QuarkShard{} bypasses the theoretical lower bounds of traditional paradigms, achieving significantly lower intra-shard and cross-shard overheads.

Specifically, in \QuarkShard{}, the system exclusively executes Quark transactions.
Each Quark transaction debits balances from only a single shard, although it may generate credits for other shards.
Because a Quark transaction has a unique input shard, its validity can be determined solely by that shard.
Consequently, processing a Quark transaction requires consensus only within its input shard; output shards do not need to execute an additional consensus round for the original transaction.
Instead, output shards simply trust the consensus result of the input shard and import the balance increments of the output accounts with minimal overhead.

\begin{figure}[t]
    \centering
    \includegraphics[width=0.47\textwidth]{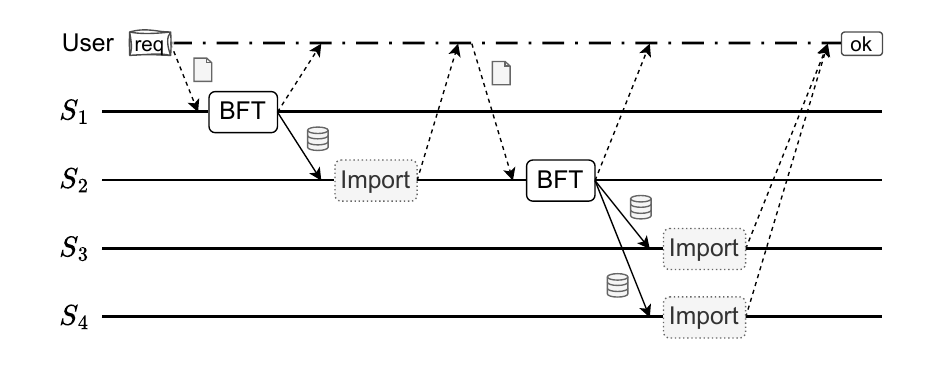}
    \caption{\QuarkShard{} Workflow. The user executes a single Quark transaction per input shard to fully process the request, incurring an IS-overhead of merely \(a\mathcal{B}\). For any Quark transaction, its generated outputs are imported by the target destination shards with negligible overhead (depicted as gray blocks), eliminating the need for an additional BFT consensus round on the original request.}
    \label{fig:quarkshardflow}
\end{figure}

\begin{figure}[t]
    \centering
    \includegraphics[width=0.47\textwidth]{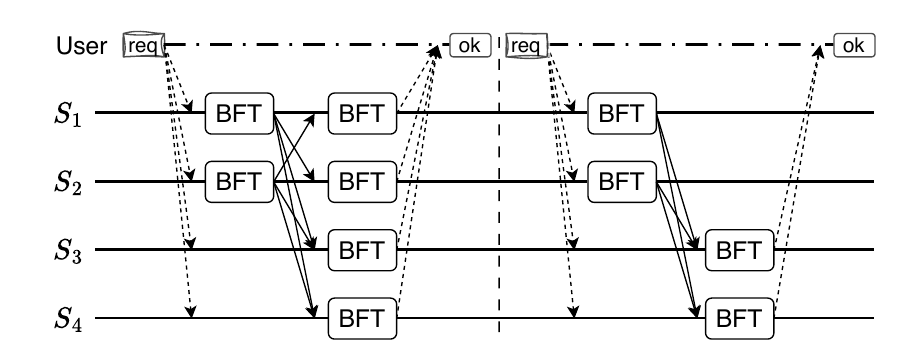}
    \caption{2PC (left) and Kronos (right) Workflow. 2PC first executes BFT consensus across all input shards to lock the required inputs; upon successful locking, it executes consensus across all input or output shards to commit input deductions or credit output balances, incurring an IS-overhead of \((2a+b)\mathcal{B}\).
    Kronos first executes BFT consensus across all input shards to consume the inputs; upon success, it executes consensus across all output shards to apply the outputs, incurring an IS-overhead of \((a+b)\mathcal{B}\).}
    \label{fig:2pcflow}
\end{figure}

To execute a cross-shard request, the user divides it into a sequence of Quark transactions.
We refer to this process as \emph{user request compilation}.
A carefully designed compilation algorithm guarantees the request's \emph{fund atomicity}: either the request executes entirely, or all funds remain under the control of their original account holders.
Furthermore, honest users can achieve a stronger guarantee, \emph{location atomicity}, which eliminates any side effects of request execution on funds controlled by honest accounts.

As an example (Figure~\ref{fig:quarkshardflow}), consider a user holding inputs in shards \(\Shard{1}\) and \(\Shard{2}\) who wishes to pay recipient accounts in \(\Shard{3}\) and \(\Shard{4}\).
The base compilation algorithm (Algorithm~\ref{alg:basic-compilation}) first issues a Quark transaction to transfer the inputs in \(\Shard{1}\) to the same account in \(\Shard{2}\).
Subsequently, it issues another Quark transaction from \(\Shard{2}\) that transfers all gathered inputs to the target accounts in \(\Shard{3}\) and \(\Shard{4}\).
On the success path, the base compilation algorithm commits one Quark transaction per input shard, incurring an IS-overhead of \(a\mathcal{B}\).
For comparison, the execution workflows of 2PC and Kronos for processing this request are illustrated in Figure~\ref{fig:2pcflow}.

Regarding CS-overhead, \QuarkShard{} requires a shard to broadcast the cross-shard outputs of Quark transactions to the target shards.
However, the cross-shard effects of a large volume of transactions can be aggregated by account.
Through this mechanism, \QuarkShard{} reduces the CS-overhead of \(r\) Quark transactions to \(O(nc\lambda)\), where \(c\) represents the number of distinct destination accounts in distinct shards across these \(r\) transactions.
By incorporating the cross-shard communication optimization from Logos~\cite{liu2026logos}, the CS-overhead can be further lowered to \(O(c\lambda)\).

The design of \QuarkShard{} not only reduces the overhead of processing cross-shard requests,
but also enhances the security and robustness of sharding blockchain systems.
In blockchain environments, adversaries can exploit transaction ordering to extract Maximal Extractable Value (MEV) at honest users' expense~\cite{daian2020flash,qin2022quantifying}.
sharding blockchains are particularly susceptible to these attacks~\cite{zhang2023front}.
\QuarkShard{} mitigates the impact of MEV attacks on sharding blockchains at two distinct levels.

First, in traditional cross-shard protocols, a cross-shard request must reach consensus in the input shards before undergoing consensus in the output shards.
This sequential workflow inherently imposes higher latency on cross-shard requests compared to intra-shard transactions, expanding the exposure to front-running attacks.
An attacker observing a published cross-shard request during its processing in the input shard can issue a front-running intra-shard transaction in the output shard, executing it before the cross-shard request completes~\cite{zhang2023front}.
In \QuarkShard{}, a user's cross-shard request is never publicly broadcast as a monolithic intent.
Consequently, when an adversary observes a user transferring funds from one shard to another, they cannot ascertain what subsequent transaction, if any, the user will initiate.
This effectively closes the extra front-running attack window caused by latency disparities between cross-shard and intra-shard execution.

Second, in traditional cross-shard protocols, achieving order-fairness across all requests requires additional cross-shard ordering mechanisms due to multi-shard involvement~\cite{zhang2023front,hong2023prophet}.
In contrast, each shard in \QuarkShard{} processes only the Quark transactions without having to coordinate with other shards, the intra-shard consensus can directly employ standard order-fairness protocols designed for non-sharding blockchains~\cite{kelkar2020order,kelkar2023themis} and ensuring fair ordering for all Quark transactions in the same shard.
It is worth noting that \QuarkShard{} does not eliminate MEV entirely; rather, it shrinks the MEV attack surface introduced by sharding and offers superior compatibility with conventional anti-MEV mechanisms.

Existing work indicates that sharding blockchains are more vulnerable to flooding attacks compared to non-sharding systems~\cite{liu2026logos}.
In most 2PC-based protocols and Kronos, a failed cross-shard request incurs a wasted IS-overhead of at least \(2x\mathcal{B}\), where \(x\) is the number of input shards that have consumed or locked funds prior to the rejection of the invalid request.
As a result, an adversary can exhaust system resources by flooding the system with invalid cross-shard requests.
Logos mitigates this by decoupling the broadcast and agreement phases of consensus, allowing invalid requests to fail early during the broadcast phase and reducing the wasted overhead to \(a\mathcal{C}\), where \(\mathcal{C}\) denotes the cost of broadcasting a transaction in consensus.\footnote{Although Logos claims an overhead of \(x\mathcal{C}\), its algorithm and security proof explicitly require all input shards to execute csdRBC, rather than only valid input shards; hence, the actual overhead is \(a\mathcal{C}\).}
In \QuarkShard{}, an invalid request incurs an overhead of \(x\mathcal{B}\) when relying on fund atomicity.
If location atomicity is required, an additional refund transaction is executed, resulting in an overhead of \((x+\mathbf{1}[x>0])\mathcal{B}\).
In scenarios where broadcast bandwidth is the primary bottleneck, \QuarkShard{} can achieve lower practical overhead than Logos.
A comparison between \QuarkShard{} and existing sharding blockchain architectures is presented in Table~\ref{tab:protocol-comparison}.

\textbf{Our contributions are summarized as follows:}
\begin{itemize}[leftmargin=*]
    \item We propose \QuarkShard{}, which breaks away from the traditional cross-shard request processing paradigm that treats each user request as a monolithic entity. By compiling and decomposing a cross-shard user request into a sequence of single-input-shard Quark transactions, \QuarkShard{} achieves an IS-overhead of \(O(a\mathcal{B})\) and an optimal CS-overhead of \(O(c\lambda)\), both below the theoretical lower bounds imposed by the traditional paradigm. Meanwhile, \QuarkShard{} enhances the security of sharding blockchain systems in aspects such as resistance to MEV and flood attacks.
    \item We present a carefully designed request compilation algorithm that decomposes a cross-shard request into multiple Quark transactions. The base algorithm guarantees fund atomicity, and its IS-overhead on the success path is \(a\mathcal{B}\). Furthermore, a refund algorithm ensures location atomicity for honest users.
    \item We implemented a prototype of \QuarkShard{} and evaluated its performance on an empirical wide-area network (WAN) spanning four AWS regions, benchmarking it against 2PC and Kronos. Experimental results demonstrate that \QuarkShard{} exhibits exceptional performance, achieving a peak throughput of $1024$\,k requests per second across $96$ nodes and $24$ shards while consistently maintaining sub-$4$-second latency. Compared to Kronos and 2PC, \QuarkShard{} achieves up to $210.7\%$ and $314.8\%$ of their maximum throughputs, respectively.
\end{itemize}

\section{Preliminaries}
\label{sec:preliminaries}

This section defines the sharding blockchain model, the user requests and Quark transactions, the building blocks of \QuarkShard{}, and the network and adversary model.

\subsection{Sharding Blockchain Model}
\label{sec:prelim-sharding-model}

The system consists of \(m\) shards, \(\setS=\{\Shard{1},\ldots,\Shard{m}\}\).
Each shard consists of \(n\) replicas, among which at most \(f\) are Byzantine, satisfying \(n\geq 3f+1\).\footnote{We assume the existence of committee assignment and reconfiguration mechanisms that ensure each shard satisfies this requirement throughout system execution with overwhelming probability~\cite{luu2016secure,kokoris2018omniledger,zamani2018rapidchain}.}
A shard \(\Shard{i}\) maintains a local chain \(\Chain{i}=(\Block{i}{0},\Block{i}{1},\ldots)\) and a replicated account state.

In blockchain systems, common account models include the account-balance model~\cite{wood2014ethereum} and the Unspent Transaction Output (UTXO) model~\cite{nakamoto2008bitcoin}.
For simplicity, this paper adopts the account-balance model, although \QuarkShard{} is equally applicable to the UTXO model.
An account is uniquely identified by its cryptographic public key \(pk\) and may hold balances across multiple shards.
Let \(\bal{i}{pk}\) denote the balance of account \(pk\) recorded in shard \(\Shard{i}\).

\subsection{User Requests and Quark Transactions}
\label{sec:prelim-user-request-quark-tx}

A \emph{user request} represents a payment intent whose inputs and outputs may span multiple shards.

\begin{definition}[User request]
\label{def:user-request}
A user request \(req\) is defined as a tuple \(\anglebk{I,O}\), where \(I\) is a set of inputs and \(O\) is a set of outputs.
Each input or output is denoted as \(\anglebk{pk,v,\Shard{i}}\),
which specifies an account \(pk\), an amount \(v\), and the shard \(\Shard{i}\) where the corresponding balance is or should be recorded.
\end{definition}

A request is considered \emph{well-formed} if the sum of its input values equals the sum of its output values, and the users authorize every input.

\QuarkShard{} does not directly process user requests; instead, it executes \emph{Quark transactions}.

\begin{definition}[Quark transaction]
\label{def:quark-transaction}
A transaction \(tx\) is defined as a Quark transaction of input shard \(\Shard{i}\) if all balances it debits are recorded in \(\Shard{i}\).
Quark transactions have the form \(\anglebk{I, O, nonce, sig}\), where \(I\) is a set of inputs and \(O\) is a set of outputs, \(nonce\) is a unique identifier for the transaction, and \(sig\) is a set of signatures for the transaction of the input accounts.
All inputs of a Quark transaction must belong to the same shard, while its outputs may belong to any shard.
\end{definition}

A Quark transaction is considered \emph{valid} in shard \(\Shard{i}\) if all inputs are in \(\Shard{i}\) and enough, the sum of its input values equals the sum of its output values, the \(nonce\) is unique, and all signatures in \(sig\) are valid for the corresponding input accounts.

To process a cross-shard request, a user \emph{compiles} it into multiple Quark transactions and submits them sequentially to their respective shards.

\begin{definition}[Request compilation algorithm]
\label{def:request-compilation}
A request compilation algorithm is a procedure that takes a well-formed user request as input, decomposes it into a sequence of Quark transactions, and submits them to their respective shards according to prescribed rules and ordering constraints.
\end{definition}

\subsection{BFT Protocol}
\label{sec:prelim-bft-proof}

Each shard uses a BFT protocol to commit a totally ordered sequence of blocks.
The BFT protocol must satisfy the following properties if at most \(f < n/3\) replicas are Byzantine:

\begin{itemize}[leftmargin=*]
    \item \textbf{Agreement.} Honest replicas never commit different blocks at the same height.
    \item \textbf{Validity.} A committed block satisfies the shard's deterministic block-validity predicate.
    \item \textbf{Liveness.} Blocks are continuously committed, with at least a fixed fraction of them being proposed by honest nodes.
\end{itemize}

\subsection{Cross-Shard Reliable Delivery}
\label{sec:prelim-csrd}

\QuarkShard{} requires reliable delivery of intra-shard consensus results from one shard to all honest nodes in another shard.
To address this requirement, we define a new cross-shard communication primitive named Cross-Shard Reliable Delivery (CSRD).

A CSRD instance from shard \(\Shard{i}\) to shard \(\Shard{j}\) is identified by a tag \(\tau\) and provides the following interface:
\begin{itemize}[leftmargin=*]
    \item \(\csrd{i}{j}.\Send{\tau,M}\), invoked by replicas in \(\Shard{i}\) to send a message \(M\);
    \item \(\csrd{i}{j}.\Deliver{\tau}\rightarrow(M,\pi)\), invoked by replicas in \(\Shard{j}\) to output the delivered message \(M\) and its corresponding certificate \(\pi\);
    \item \(\csrd{i}{j}.\Verify{\tau,\pi}\), which publicly verifies the certificate \(\pi\) without obtaining \(M\).
\end{itemize}

CSRD must satisfy the following properties if at most \(f < n/3\) replicas are Byzantine in each shard:

\begin{itemize}[leftmargin=*]
    \item \textbf{Correctness.} If all honest source replicas invoke \(\send\) on the same message \(M\), then all honest destination replicas will eventually obtain \(M\) and a certificate \(\pi\) via \(\deliver\), such that \(\pi\) is accepted by \(\verify\).
    \item \textbf{Validity.} If any replica obtains a certificate \(\pi\) accepted by \(\verify\), then at least one honest source replica must have invoked \(\send\).
\end{itemize}

In Section~\ref{app:csrd}, we present a concrete construction of CSRD that leverages threshold signatures, cryptographic commitments, and erasure codes (EC).

\subsection{Network and Adversary Model}
\label{sec:prelim-net-adversary}

The network model can be either asynchronous or partially synchronous; the underlying network assumptions can be adapted by adjusting the intra-shard BFT protocol and the CSRD primitive used to construct \QuarkShard{}.
All messages are transmitted over authenticated channels, and eventual message delivery is guaranteed.
In an asynchronous network, message propagation delays can be arbitrarily large.
In a partially synchronous network, there exists a Global Stabilization Time (GST) such that, after GST, all messages transmitted between honest nodes are guaranteed to be delivered within a known bounded delay.

The adversary may corrupt and fully control up to \(f < n/3\) replicas per shard, causing them to arbitrarily deviate from the protocol.
Subject to the constraints of the network model, the adversary can arbitrarily delay or reorder messages, but cannot forge or suppress messages sent by other replicas.

\section{\QuarkShard{} Protocol}
\label{sec:protocol}

In this section, we present the details of \QuarkShard{}, with its primary procedures outlined in Algorithm~\ref{alg:quark}.
The protocol reaches consensus on Quark transactions within their input shards, leverages the CSRD primitive to deliver cross-shard account credits to output shards, and imports these account credits into output shards with minimal overhead.
This section focuses on the execution and processing of Quark transactions, while the user request compilation scheme is detailed in Section~\ref{sec:userlevel}.

\subsection{Block Structure and Execution}

In \QuarkShard{}, each shard runs an intra-shard BFT consensus protocol to agree on a chain of blocks, denoted as \(\Chain{i}=(\Block{i}{0},\Block{i}{1},\ldots)\).
Each block takes the form \(\anglebk{meta,F,txs}\), where \(meta\) includes the parent block hash, state root, transaction root, and other standard header fields, \(F\) denotes the block's import frontier (detailed below), and \(txs\) is an ordered sequence of Quark transactions contained within the block.

Upon block commitment (line~\ref{alg:quark:committing:begin}--\ref{alg:quark:committing:end}), replicas execute all transactions contained therein.
For intra-shard inputs and outputs targeting the current shard, the corresponding account balances are directly debited or credited.
For cross-shard outputs targeting external shards, replicas aggregate their credits into a \emph{cross-shard credit batch}, which is subsequently transmitted to the destination shard via the CSRD primitive after block execution.
Specifically, all account credits generated by shard \(\Shard{i}\) at block height \(h\) and destined for shard \(\Shard{j}\) are aggregated into a batch \(\DeltaSet{i}{j}{h}\), where \(\DeltaSet{i}{j}{h}[pk]\) represents the balance credit for account \(pk\) in the destination shard.

For clarity of presentation and proof, we require that even empty credit batches be transmitted to destination shards; we present a straightforward optimization in Section~\ref{sec:discussion:emptybatch} that eliminates the overhead incurred by empty credit batches.

Destination shards store incoming credit batches in a local credit pool and import them into subsequent blocks (line~\ref{alg:quark:csrd:begin}--\ref{alg:quark:csrd:end}).
Because credit batches are uniquely produced through input-shard consensus and reliably delivered via CSRD, destination shards do not need to execute an additional consensus round on the contents of each credit batch.
Instead, each block records an import frontier vector \(F\) of length equal to the number of shards, where
\begin{align*}
\Block{j}{h}.F[i]=&\max\{h':\forall r\in[0,h'),\\
&\DeltaSet{i}{j}{r}\text{ has been imported by the end of }\Block{j}{h}\}.
\end{align*}
\noindent For convenience, we define \(\Block{j}{-1}.F[i]=0\) and \(\Block{j}{h}.F[j]=0\).

Upon block commitment, replicas wait until all credit batches selected for import have arrived, and apply the corresponding balance credits to their respective accounts before executing any transactions.

\subsection{Block Proposal and Validation}

\begin{algorithm}[htbp]
\footnotesize
\DontPrintSemicolon
\AlgTwoERole{\textbf{as a replica in shard} \(\Shard{i}\)}
\AlgUpon{receiving a Quark transaction \(tx\) from user\label{alg:quark:receiving:begin}}{
    \eIf{\(tx\) is valid}{
        Store \(tx\) into the transaction pool.\;
    }{
        Response to the user that \(tx\) is rejected.\;
    }\label{alg:quark:receiving:end}
}
\AlgTwoESep
\AlgUpon{proposing a new block \(\Block{i}{h}\)\label{alg:quark:proposing:begin}}{
    Let \(F\gets\Block{i}{h-1}.F\); \(W\gets\varnothing\).\;
    \For{each source \(\Shard{k}\ne\Shard{i}\)}{
        Let \(h_k\) be the largest value such that for all \(r\in[F[k],h_k)\), \((\DeltaSet{k}{i}{r},\DeltaProof{k}{i}{r})\) exists in the credit pool.\;
        For all \(r\in[F[k],h_k)\), add \(\DeltaProof{k}{i}{r}\) to \(W\).\;
        Let \(F[k]\gets h_k\).\;
    }
    Let \(\Block{i}{h}.F\gets F\).\;
    Select valid Quark transactions in order from the transaction pool and add them to \(\Block{i}{h}.txs\).\;
    Fill in \(\Block{i}{h}.meta\).\;
    Propose \(\langle\Block{i}{h},W\rangle\).\;\label{alg:quark:proposing:end}
}
\AlgTwoESep
\AlgUpon{verifying proposal \(\langle\Block{i}{h},W\rangle\)\label{alg:quark:verifying:begin}}{
    \If{any field in \(\Block{i}{h}.meta\) except the state root is invalid}{
        Refuse to sign for \(\Block{i}{h}\).\;
    }
    \If{\(\Block{i}{h}.F[k]<\Block{i}{h-1}.F[k]\) for any \(k\)}{
        Refuse to sign for \(\Block{i}{h}\).\;
    }
    \For{each \(k\ne i\) and \(r\in[\Block{i}{h-1}.F[k],\Block{i}{h}.F[k])\)}{
        \If{\(\DeltaProof{k}{i}{r}\not\in W\) or \(\csrd{k}{i}.\Verify{r,\DeltaProof{k}{i}{r}}\) is false}{
            Refuse to sign for \(\Block{i}{h}\).\;
        }
    }
    \For{each \(k\ne i\) and \(r\in[\Block{i}{h-1}.F[k],\Block{i}{h}.F[k])\)}{
        \WaitUntil{} \(\DeltaSet{k}{i}{r}\) appears in the credit pool.\;
    }
    \If{not all transactions in \(\Block{i}{h}.txs\) are valid}{
        Refuse to sign for \(\Block{i}{h}\).\;
    }
    \If{the state root in \(\Block{i}{h}\) does not match the root of the state obtained after importing all selected credit batches and executing all transactions}{
        Refuse to sign for \(\Block{i}{h}\).\;
    }
    Sign for \(\Block{i}{h}\) if all checks pass.\;\label{alg:quark:verifying:end}
}
\AlgTwoESep
\AlgUpon{committing \(\Block{i}{h}\)\label{alg:quark:committing:begin}}{
    \For{each source \(\Shard{k}\ne\Shard{i}\)}{
        \For{each \(r\in[\Block{i}{h-1}.F[k],\Block{i}{h}.F[k])\) in increasing order}{
            \WaitUntil{} \(\DeltaSet{k}{i}{r}\) appears in the credit pool.\;
            Add every credit in \(\DeltaSet{k}{i}{r}\) to the target accounts.\;
            Delete \((\DeltaSet{k}{i}{r},\DeltaProof{k}{i}{r})\) from the credit pool.\;
        }
    }
    Initialize \(\DeltaSet{i}{j}{h}\gets\varnothing\) for every \(\Shard{j}\ne\Shard{i}\).\;
    \For{each \(tx\in\Block{i}{h}.\mathsf{txs}\) in order}{
        Debit its inputs.\;
        \For{each output \(\anglebk{pk,v,\Shard{j}}\)}{
            \eIf{\(j=i\)}{
                \(\bal{i}{pk}\gets\bal{i}{pk}+v\).\;
            }{
                \(\DeltaSet{i}{j}{h}[pk]\gets\DeltaSet{i}{j}{h}[pk]+v\).\;
            }
        }
    }
    \For{each \(\Shard{j}\ne\Shard{i}\)}{
        Invoke \(\csrd{i}{j}.\Send{h,\DeltaSet{i}{j}{h}}\).\;
    }\label{alg:quark:committing:end}
}
\AlgTwoESep
\AlgTwoERole{\textbf{as a replica in shard} \(\Shard{j}\)}
\AlgUpon{\(\csrd{i}{j}.\Deliver{h}\rightarrow(\DeltaSet{i}{j}{h},\DeltaProof{i}{j}{h})\)\label{alg:quark:csrd:begin}}{
    Store \((\DeltaSet{i}{j}{h},\DeltaProof{i}{j}{h})\) into the credit pool.\;\label{alg:quark:csrd:end}
}
\caption{\QuarkShard{} Block Processing.}
\label{alg:quark}
\end{algorithm}

Because each Quark transaction accesses a unique input shard, upon receiving a user transaction, replicas in the input shard can determine its validity solely based on their local state.
If valid, the transaction is added to the local transaction pool; otherwise, the replica rejects it and notifies the user (line~\ref{alg:quark:receiving:begin}--\ref{alg:quark:receiving:end}).
If a transaction already in the pool becomes invalid due to subsequent conflicting transactions, replicas remove it and inform the user.
Note that when evaluating transaction validity, replicas account for both current account balances and pending balance credits residing in the credit pool.

When proposing a new block (line~\ref{alg:quark:proposing:begin}--\ref{alg:quark:proposing:end}), a replica first inspects the parent block's import frontier and the current credit pool.
For each input shard, the replica selects the highest contiguous height at which all corresponding credit batches are present, setting this value as the new block's import frontier.
The replica then selects a sequence of valid Quark transactions from the transaction pool to populate the new block, where balance sufficiency checks account for the newly imported funds in this block.
Finally, the replica computes the block metadata.

Along with the proposed block, the replica broadcasts a validation witness \(W\), which consists of cryptographic proofs for all newly imported credit batches.
The witness \(W\) serves as proof that at least one honest node in the source shard has transmitted the corresponding credit batch, ensuring that any replica can safely await the required credit batches without incurring liveness issues.
The witness \(W\) needs to be broadcast only once during block proposal; it is neither subjected to consensus nor persistently stored.

Upon receiving a block proposal (line~\ref{alg:quark:verifying:begin}--\ref{alg:quark:verifying:end}), a replica first verifies the correctness of its metadata (excluding the state root).
It then verifies that the block's import frontier is non-decreasing relative to its parent, and that all necessary cryptographic proofs in \(W\) are valid. Once these checks pass, the replica can safely wait for all required credit batches to arrive.
Next, the replica sequentially verifies the validity of all block transactions after applying the newly imported funds, and confirms that the resulting state root matches the proposed state root.
Replicas cast their consensus votes for the block if and only if all validation checks succeed.

The aforementioned validation procedure ensures that credit batch importation is monotonic and contiguous, with each height's credit batch being imported exactly once.
Furthermore, it guarantees that all transactions contained within committed blocks are valid Quark transactions.

\subsection{CSRD}

\QuarkShard{} employs the CSRD primitive to guarantee that credit batches derived from an input shard are reliably and verifiably delivered to target output shards.
CSRD can be realized through various implementations; here, we briefly outline one concrete construction, with complete structural details deferred to Appendix~\ref{app:csrd:instance}.
Furthermore, by integrating the RPST primitive proposed by Logos~\cite{liu2026logos}, the vast majority of cross-shard communication overhead can be shifted onto intra-shard networks, as elaborated in Appendix~\ref{app:csrd:logos}.

Specifically, to verifiably deliver a message \(M\) bound to a tag \(\tau\), each replica in the input shard generates a threshold signature on \(\tau\) with a threshold of \(n-f\), and broadcasts this signature to all nodes in the target output shard.
Upon collecting \(n-f\) distinct valid signatures, a node in the target shard aggregates them into a valid cryptographic certificate.
This certificate provides undeniable proof that at least \(n-2f\) honest replicas in the source shard authorized and transmitted the message associated with \(\tau\).

To minimize communication cost when transmitting the payload \(M\), each replica in the input shard applies erasure code (EC) to \(M\) with a reconstruction threshold of \(f+1\).
Each replica then broadcasts its designated EC fragment corresponding to its replica index along with its threshold signature share.
Once a target shard node receives \(n-f\) valid signatures and constructs the certificate, it is guaranteed that at least \(n-2f \geq f+1\) honest replicas have contributed valid EC fragments.
Consequently, these fragments suffice to reconstruct the original message \(M\).
To prevent Byzantine nodes from broadcasting valid signature shares alongside corrupted EC fragments, the scheme secures all fragments using cryptographic commitments.

\section{User Request Compilation}
\label{sec:userlevel}

To process a request spanning multiple input and output shards, the user employs a request compilation algorithm to decompose it into a sequence of Quark transactions, which are subsequently submitted to their corresponding shards according to prescribed rules.
In this section, we first formally define two notions of atomicity for user request processing: \emph{fund atomicity} and \emph{location atomicity}.
We then present a base compilation algorithm that achieves fund atomicity, followed by a refund mechanism that strengthens this guarantee to location atomicity.
Finally, we introduce a prefunding optimization that reduces user request overhead and latency.

\subsection{Atomicity}

We observe that system safety is preserved for an execution as long as all funds remain under the control of their original account holders upon failure.
Therefore, we first define a safe, albeit relaxed, property compared to conventional atomicity~\cite{kokoris2018omniledger,al2017chainspace,dang2019towards}.
Algorithms satisfying this relaxed property guarantee user fund safety while incurring lower overhead along failure execution paths.

\begin{definition}[Fund Atomicity]
    \label{def:fund-atomicity}
    A request compilation algorithm satisfies fund atomicity if it guarantees that,
    either all inputs and outputs of a user request are successfully executed,
    or all involved funds remain under the control of their original account holders.
\end{definition}

Additionally, we define a stronger property that is equivalent to traditional atomicity from the perspective of honest users:

\begin{definition}[Location Atomicity]
    \label{def:loc-atomicity}
    A request compilation algorithm satisfies location atomicity if it guarantees that,
    either all inputs and outputs of a user request are successfully executed,
    or all funds in accounts controlled by honest users remain in their initial states, exactly as if the user request had never been executed.
\end{definition}

\subsection{Base Compilation and Refund}

The base compilation algorithm operates by first consolidating all required funds onto a single coordinator shard, and subsequently committing a single Quark transaction on that coordinator shard to disburse all funds simultaneously to their final destinations, thereby fulfilling the user request.
Algorithm~\ref{alg:basic-compilation} details this procedure.
Note that terms such as "funding transaction" and "final transaction" in the algorithm are introduced solely for ease of exposition; they represent standard, independent Quark transactions with no special properties.

In the base compilation algorithm, one of the input shards is designated as the coordinator shard.
Let \(\mathsf{InShards}(req)\) denote the set of input shards for a given user request \(req\).
For each shard in \(\mathsf{InShards}(req)\) other than the coordinator shard, the algorithm submits a funding Quark transaction.
The input of this transaction consists of all inputs of the user request on that specific shard, while its output targets the coordinator shard, keeping the account and balance amounts unchanged.
If the user request involves only a single input shard, this step is bypassed.

Subsequently, the algorithm waits until all funding transactions are committed and the credit batches containing their outputs are imported by the coordinator shard.
Once all funding transactions are finalized, a final Quark transaction is submitted to the coordinator shard.
Its inputs comprise all request inputs (now located on the coordinator shard), and its outputs correspond to all outputs specified by the user request.
The user request completes successfully once the final transaction is committed and its resulting credit batches are imported by the destination shards.

If any transaction is rejected during this process, the algorithm immediately terminates and returns a failure status.
Because the base compilation algorithm does not alter fund ownership prior to the final transaction, it inherently guarantees fund atomicity.

When an execution fails and an honest user seeks to achieve location atomicity, they can invoke the refund mechanism outlined in Algorithm~\ref{alg:location-fallback} to relocate their funds back to their original shards.
Algorithm~\ref{alg:location-fallback} takes a set of funding transactions to be reverted along with the coordinator shard as inputs, and submits a single refund transaction on the coordinator shard to transfer all migrated funds back to their initial shards.
Because an honest user will not issue conflicting transactions against the refund transaction, this procedure is guaranteed to succeed.

\begin{algorithm}[t]
\footnotesize
\DontPrintSemicolon
\KwIn{Request \(req=\anglebk{I,O}\) and coordinator \(\Shard{c}\in\mathsf{InShards}(req)\).}
\KwOut{\(\mathsf{success}\) only after all requested effects settle.}
\For{each \(\Shard{i}\in\mathsf{InShards}(req)\setminus\{\Shard{c}\}\) \textbf{in parallel}}{
    \AlgConstruct{a funding Quark transaction \(tx_i\) in \(\Shard{i}\)}{
        \(tx_i.I=\{\anglebk{pk,v,\Shard{i}}\mid \anglebk{pk,v,\Shard{i}}\in req.I\}\)\;
        \(tx_i.O=\{\anglebk{pk,v,\Shard{c}}\mid \anglebk{pk,v,\Shard{i}}\in req.I\}\)\;
    }
    Submit \(tx_i\) to \(\Shard{i}\).\;
}
\WaitUntil{} every funding transaction is committed or rejected.\;
\If{a funding transaction was rejected}{
    \Return \(\mathsf{failure}\).\;
}
\WaitUntil{} every corresponding credit batch is imported by \(\Shard{c}\).\;
\AlgConstruct{a final Quark transaction \(tx_f\) in \(\Shard{c}\)}{
    \(tx_f.I=\{\anglebk{pk,v,\Shard{c}}\mid\anglebk{pk,v,\Shard{k}}\in req.I\}\)\;
    \(tx_f.O=req.O\)\;
}
Submit \(tx_f\) to \(\Shard{c}\).\;
\WaitUntil{} the final transaction is committed or rejected.\;
\If{\(tx_f\) was rejected}{
    \Return \(\mathsf{failure}\).\;
}
\WaitUntil{} every corresponding credit batch of \(tx_f\) is imported.\;
\Return \(\mathsf{success}\).\;
\caption{Base Compilation of a User Request.}
\label{alg:basic-compilation}
\end{algorithm}

\begin{algorithm}[t]
\footnotesize
\DontPrintSemicolon
\KwIn{A transaction set \(TX\) to be refunded, coordinator \(\Shard{c}\).}
\Return if \(TX=\varnothing\).\;
\WaitUntil{} every corresponding credit batch of \(tx\in TX\) is imported.\;
\AlgConstruct{a refund Quark transaction \(tx_r\) in \(\Shard{c}\)}{
    \(tx_r.I=\{\anglebk{pk,v,\Shard{c}}\mid \anglebk{pk,v,\Shard{c}}\in tx_i.O\text{ and }tx_i\in TX\}\)\;
    \(tx_r.O=\{\anglebk{pk,v,\Shard{i}}\mid \anglebk{pk,v,\Shard{i}}\in tx_i.I\text{ and }tx_i\in TX\}\)\;
}
Submit \(tx_r\) to \(\Shard{c}\).\;
\WaitUntil{} the refund transaction is committed.\;
\WaitUntil{} every corresponding credit batch of \(tx_r\) is imported.\;
\caption{Refund for Location Atomicity.}
\label{alg:location-fallback}
\end{algorithm}

\subsection{Prefunding}
\label{sec:userlevel:prefunding}

In the base compilation algorithm, users must transfer funds to the coordinator shard prior to executing the final transaction.
However, users are not required to execute the base compilation independently for every individual request.
Consider a scenario where a user anticipates conducting multiple transactions with a specific coordinator shard, while holding the majority of their funds on another shard.
In this case, the user can prefund a sufficient balance to the coordinator shard in advance through a single batch transfer.
During subsequent request executions, the user bypasses the need to submit fund-transferring funding transactions, directly issuing the final transaction instead.
Consequently, the overhead incurred by multiple funding transactions is effectively amortized while simultaneously reducing request execution latency.

For clarity of presentation, transaction fee mechanisms are not explicitly detailed in this paper.
However, in typical blockchain systems, users generally pay execution fees for each transaction~\cite{wood2014ethereum}.
Under such economic models, users are inherently motivated to minimize the number of submitted transactions required to process a given volume of requests, which directly aligns with our design goal of reducing overall system overhead.

\section{Security Analysis}
\label{sec:security}

In this section, we formally prove the security properties of \QuarkShard{}.
We first establish the safety and liveness of \QuarkShard{} when processing Quark transactions.
Next, we demonstrate that Algorithm~\ref{alg:basic-compilation} and Algorithm~\ref{alg:location-fallback} satisfy fund atomicity and location atomicity.
Finally, we analyze the security guarantees provided by \QuarkShard{} regarding MEV resistance.

\subsection{Consistency, Validity, and Liveness}

\begin{theorem}[Consistency]
Any two honest replicas within the same shard commit the same block at the same height.
\end{theorem}

\begin{proof}
Since the number of Byzantine replicas \(f\) within the shard satisfies \(3f+1 \leq n\), the underlying BFT consensus protocol satisfies Agreement. Consequently, every honest replica in the shard commits the identical block at any given height.
\end{proof}

\begin{theorem}[Validity]
Any Quark transaction contained in a committed block is valid.
Any credit batch imported by a committed destination block is derived from a committed source block, and every credit batch is imported by its destination shard exactly once.
\end{theorem}

\begin{proof}
Since the number of Byzantine replicas \(f\) within the shard satisfies \(3f+1 \leq n\), the underlying BFT protocol satisfies Validity.
That is, every committed block must pass the verification rules specified in Algorithm~\ref{alg:quark} lines~\ref{alg:quark:verifying:begin}--\ref{alg:quark:verifying:end}.

During block verification, the validity of all transactions contained within the block is enforced.
Therefore, all Quark transactions included in a committed block are valid.

Similarly, verification ensures that the proof for each imported credit batch exists and is valid.
By the Validity property of CSRD, at least one honest replica in the source shard has transmitted the credit batch.
According to Algorithm~\ref{alg:quark}, an honest replica transmits the credit batch generated by a block if and only if that block has been committed.
Thus, any credit batch imported by a committed block is guaranteed to originate from a committed block.

Furthermore, according to the verification rules, the importation of credit batches is strictly monotonic and contiguous.
Consequently, each credit batch is imported by its destination shard exactly once.
\end{proof}

\begin{theorem}[Liveness]
Any Quark transaction added to an honest replica's transaction pool that remains persistently valid will eventually be included in a committed block.
Every generated credit batch will eventually be imported by its destination shard.
\end{theorem}

\begin{proof}
During block verification and commitment, \QuarkShard{} waits for all required credit batches to arrive; we first prove that this waiting mechanism does not compromise the Liveness property of the underlying BFT protocol.
During verification, an honest replica waits for a credit batch only after validating that its cryptographic proof exists and is valid.
By the Validity of CSRD, at least one honest replica in the source shard has transmitted this batch.
Hence, the block deriving this credit batch has already been committed.
According to Algorithm~\ref{alg:quark}, all honest replicas invoke CSRD to transmit the credit batch upon block commitment.
By the Correctness of CSRD, all honest replicas in the destination shard will eventually receive the batch.
Thus, waiting for credit batches during verification does not undermine liveness.
Analogously, because a committed block must be a verifiable block, waiting for credit batches during commitment also preserves liveness.

Since the number of Byzantine replicas \(f\) within the shard satisfies \(3f+1 \leq n\), the running BFT consensus protocol satisfies Liveness.
Consequently, blocks proposed by honest replicas are continuously committed.
According to the proposal rules for honest replicas (Algorithm~\ref{alg:quark} lines~\ref{alg:quark:proposing:begin}--\ref{alg:quark:proposing:end}), honest replicas propose valid Quark transactions sequentially and import as many pending credit batches as possible.
Therefore, any Quark transaction added to an honest replica's transaction pool that remains continuously valid will eventually be included in a committed block.

For any credit batch derived from a committed block, all honest replicas in the destination shard will eventually receive it pursuant to honest transmission rules and the Correctness of CSRD.
Therefore, all generated credit batches will eventually be imported by their respective destination shards.
\end{proof}

\subsection{Atomicity}

\begin{theorem}[Fund Atomicity]
Algorithm~\ref{alg:basic-compilation} satisfies fund atomicity.
\end{theorem}
\begin{proof}
First, consider the case where Algorithm~\ref{alg:basic-compilation} completes successfully.
In this scenario, all funding transactions and the final transaction are committed, and their corresponding credit batches are imported.
The funding transactions migrate all inputs of the user request to the coordinator shard, whereupon the final transaction transfers these funds to all destination outputs specified by the user request.
Thus, along the execution success path, all inputs and outputs of the user request are fully executed.

Conversely, when the algorithm fails, the final transaction is necessarily not executed.
In this case, only a subset of funding transactions may have been executed.
Because funding transactions do not alter the ownership of the funds, all involved funds remain under the control of their original accounts.
\end{proof}

\begin{theorem}[Location Atomicity]
Upon failure of Algorithm~\ref{alg:basic-compilation}, an honest user can invoke Algorithm~\ref{alg:location-fallback} to guarantee location atomicity.
\end{theorem}
\begin{proof}
According to the execution logic of Algorithm~\ref{alg:location-fallback}, upon successful execution, all funds associated with the reverted funding transactions are transferred back to their original shards.
Furthermore, because honest users do not issue conflicting transactions against the refund transaction, Algorithm~\ref{alg:location-fallback} is guaranteed to execute successfully.
Consequently, location atomicity is preserved.
\end{proof}

\subsection{Anti-MEV}

\begin{definition}[Cross-Shard Front-Running Free]
\label{def:CSFFF}
A sharding blockchain system is cross-shard front-running free if executing a front-running attack against a cross-shard request is no easier for an adversary than executing a front-running attack against an intra-shard request.
\end{definition}

\begin{theorem}[Cross-Shard Front-Running Free]
~\\
\QuarkShard{} is cross-shard front-running free.
\end{theorem}

\begin{proof}
In \QuarkShard{}, a cross-shard user request is compiled and decomposed into a sequence of individual, independent Quark transactions.
An adversary observing an issued Quark transaction cannot infer the structure of the underlying complete user request or predict subsequent Quark transactions that the user will execute.
Consequently, the adversary can initiate a front-running attack only after the specific target Quark transaction is publicly broadcast, which presents the same level of difficulty as front-running a purely intra-shard request.
\end{proof}

\begin{definition}[Order-Fairness BFT Compatibility]
\label{def:OFBC}
~\\
A sharding blockchain system is order-fairness BFT compatible if it can seamlessly integrate non-sharding order-fairness BFT consensus protocols~\cite{kelkar2020order,kelkar2023themis} within each shard to enforce order fairness for all processed transactions.
\end{definition}

\begin{theorem}[Order-Fairness BFT Compatibility]
\QuarkShard{} is Order-Fairness BFT Compatible.
\end{theorem}

\begin{proof}
In \QuarkShard{}, each shard processes only those Quark transactions for which it serves as the input shard, operating without the need for inter-shard coordination to process a given individual transaction.
As a result, Order-Fairness BFT protocols designed for non-sharding blockchains can be directly deployed within each shard to enforce order fairness across all Quark transactions.
\end{proof}

\section{Performance Analysis}
\label{sec:performance}

In this section, we analyze the performance of \QuarkShard{}.
Regarding intra-shard overhead, we first analyze the execution cost of the base compilation algorithm on the success path, followed by an analysis of the overhead resulting from invalid requests and credit batch importations.
Next, we demonstrate that \QuarkShard{} achieves exceptionally low cross-shard communication overhead.
Finally, we evaluate the processing latency of user requests under \QuarkShard{}.

\subsection{Intra-Shard Overhead}

\begin{theorem}[IS-Overhead of Base Compilation on the Success Path]
For a user request spanning \(a\) input shards, a successful execution of Algorithm~\ref{alg:basic-compilation} without prefunding incurs an IS-overhead of \(a\mathcal{B}\).
\end{theorem}

\begin{proof}
On the success path, Algorithm~\ref{alg:basic-compilation} commits one funding transaction in each of the \(a-1\) input shards other than the coordinator shard, followed by one final transaction in the coordinator shard.
Each of these Quark transactions is processed by BFT consensus within its input shard, incurring an IS-overhead of \(\mathcal{B}\) per transaction under our cost model.
Thus, the total IS-overhead is \((a-1)\mathcal{B}+\mathcal{B}=a\mathcal{B}\).
\end{proof}

Through the prefunding optimization introduced in Section~\ref{sec:userlevel:prefunding}, assuming that a single prefunding operation can be amortized across \(l\) subsequent transactions on average, the overhead of funding transactions for user requests is effectively reduced.
On the success path, the amortized IS-overhead per user request becomes \([(a-1)/l+1]\mathcal{B}\).

\begin{theorem}[IS-Overhead of Base Compilation on the Failure Path]
For an invalid user request, let \(x<a\) denote the number of funding transactions that are committed on the failure path.
When Algorithm~\ref{alg:basic-compilation} fails, the resulting IS-overhead is \(x\mathcal{B}\).
If the user requires location atomicity, executing Algorithm~\ref{alg:location-fallback} incurs an additional overhead, yielding a total IS-overhead of \((x+\mathbf{1}[x>0])\mathcal{B}\), where \(\mathbf{1}[x>0]\) is the indicator function that equals 1 if \(x>0\) and 0 otherwise.
\end{theorem}

\begin{proof}
Along the failure path, if \(x\) transactions are committed, input shards should execute BFT consensus for \(x\) transactions, thus the IS-overhead is \(x\mathcal{B}\).
If the user requires location atomicity and there are any transactions to be refunded, executing Algorithm~\ref{alg:location-fallback} submits one additional refund transaction on the coordinator shard, bringing the total IS-overhead to \((x+\mathbf{1}[x>0])\mathcal{B}\).
\end{proof}

\begin{theorem}[Overhead of Importing Credit Batches]
For a block importing \(\nu_B\) credit batches, the logical intra-shard broadcast overhead incurred by importing credit batches is \(mw + \nu_B s_\pi\), where \(m\) denotes the number of shards, \(w\) is the byte size of an integer, and \(s_\pi\) is the size of a CSRD-generated certificate.
The vast majority of this data requires only broadcast, with only the \(mw\) portion participating in the consensus protocol.
\end{theorem}

\begin{proof}
Each block stores an import frontier consisting of an integer for each shard to participate in consensus, incurring a total consensus overhead of \(mw\).
Furthermore, during block proposal, the replica broadcasts cryptographic proofs for all \(\nu_B\) imported credit batches in the block, contributing an additional broadcast overhead of \(\nu_B s_\pi\).
Consequently, the total intra-shard broadcast overhead sums to \(mw + \nu_B s_\pi\).
\end{proof}

Assuming a standard integer size of approximately 8~Bytes and a Threshold BLS (TBLS) signature certificate size of approximately 48~Bytes, in a system with 100 shards where all shards produce blocks at comparable rates and each block contains cross-shard outputs destined for every other shard, the importation overhead is approximately 5.42~KB.
This accounts for merely 0.26\% of a classic 2~MB block in Bitcoin.
In practice, because empty credit batches can be compressed using the scheme presented in Section~\ref{sec:discussion:emptybatch}, the actual overhead is even lower.

\subsection{Cross-Shard Overhead}

\begin{theorem}[CS-Overhead of \QuarkShard{}]
For a batch of \(r\) Quark transactions involving \(c\) distinct output accounts in distinct shards, the total logical data payload delivered by \QuarkShard{} to one output shard via the CSRD primitive is \(O(c\lambda)\).
Assuming the batch is sufficiently large (Lemma~\ref{lem:csrd-overhead}).
When employing the baseline CSRD implementation in Algorithm~\ref{alg:csrd}, the resulting CS-overhead is \(O(nc\lambda)\).
By using the optimized CSRD scheme in Algorithm~\ref{alg:csrd-logos}, the majority of cross-shard communication is offloaded to intra-shard networks, yielding a CS-overhead of \(O(c\lambda)\).
\end{theorem}

\begin{proof}
According to Algorithm~\ref{alg:quark}, the data delivered by \QuarkShard{} to destination shards via CSRD consists of output funds aggregated on a per-account basis across the \(r\) transactions.
For each involved account, the payload contains only the account identifier and an integer representing the net balance update.
Consequently, the total logical data payload delivered across shards is \(O(c\lambda)\).

By Lemma~\ref{lem:csrd-overhead}, using Algorithm~\ref{alg:csrd} as the concrete CSRD implementation incurs a CS-overhead of \(O(nc\lambda)\).
Furthermore, adopting Algorithm~\ref{alg:csrd-logos} reduces the CS-overhead to \(O(c\lambda)\).
\end{proof}

Note that if each transaction specifies multiple outputs, the number of output accounts may exceed the number of transactions.
However, rather than transmitting only a transaction identifier, traditional protocols are required to transmit the full transaction payload to the output shards.
In contrast, \QuarkShard{} obviates the need to transmit transaction contents by transmitting only the aggregated account credits, which strictly outperforms transaction-granularity transmission.

\subsection{Latency}
\label{sec:performance-latency}

Below, we analyze the processing latency of \QuarkShard{} for individual Quark transactions and full user requests.
For ease of analysis, let \(\BFTLatency\) bound the time elapsed from when a transaction is submitted to an honest node in a shard until it is committed via BFT consensus,
let \(\CSRDLatency\) bound the delivery latency of the CSRD primitive,
and let \(\ImportLatency\) bound the delay between a credit batch being delivered to a shard and its eventual importation into a committed block.

The processing of a single Quark transaction requires committing the transaction in its input shard and subsequently waiting for the credit batches containing its outputs to be delivered and imported by the target shard.
Its total processing latency is bounded by
\[L_{\mathrm{one}} \leq \BFTLatency + \CSRDLatency + \ImportLatency.\]
The same upper bound applies to user requests that involve only a single input shard, or whose inputs have all been prefunded to the coordinator shard.

For user requests spanning multiple input shards, the base compilation algorithm requires submitting two sequential rounds of Quark transactions, yielding a total latency bounded by
\[L_{\mathrm{base}} \leq 2(\BFTLatency + \CSRDLatency + \ImportLatency).\]

Although the latency of \QuarkShard{} when handling multi-input requests is higher than that of protocols such as 2PC (which typically require two rounds of BFT consensus and one round of cross-shard communication), this trade-off yields significantly better system throughput and scalability, while the latency can be effectively mitigated through mechanisms such as prefunding.

\section{Evaluation}
\label{sec:evaluation}

We implemented a prototype of \QuarkShard{} and deployed it on Amazon Web Services (AWS) to evaluate its empirical performance under real-world wide-area network (WAN) conditions.
The deployment spanned four AWS geographical regions (Virginia, California, Tokyo, and London), scaling up to $96$ replicas across $24$ shards.
To ensure a fair evaluation, we implemented 2PC and Kronos within the exact same experimental framework and benchmarked \QuarkShard{} against them.

\paragraph{Implementation Details.}
We implemented the core benchmarking framework as well as \QuarkShard{}, 2PC, and Kronos in Golang.
Cryptographic primitives include SHA-256 for cryptographic hashing and Threshold BLS (TBLS) for threshold signatures~\cite{boneh2001short,boldyreva2002threshold}.
Erasure coding is implemented using Reed-Solomon codes~\cite{reed1960polynomial}.
Intra-shard consensus is implemented using HotStuff~\cite{yin2019hotstuff}.

\paragraph{Deployment Details.}
Each replica was deployed on a dedicated AWS EC2 \texttt{c5.2xlarge} instance equipped with 8 vCPUs and 16~GB of RAM.
The outbound network bandwidth for each instance was capped at 100~Mbps.
The payload size for an individual transaction was set to 128~bytes, matching the average size of a standard native transfer transaction in Ethereum.

\subsection{Overall Performance}

\begin{figure}[t]
    \centering
    \includegraphics[width=0.47\textwidth]{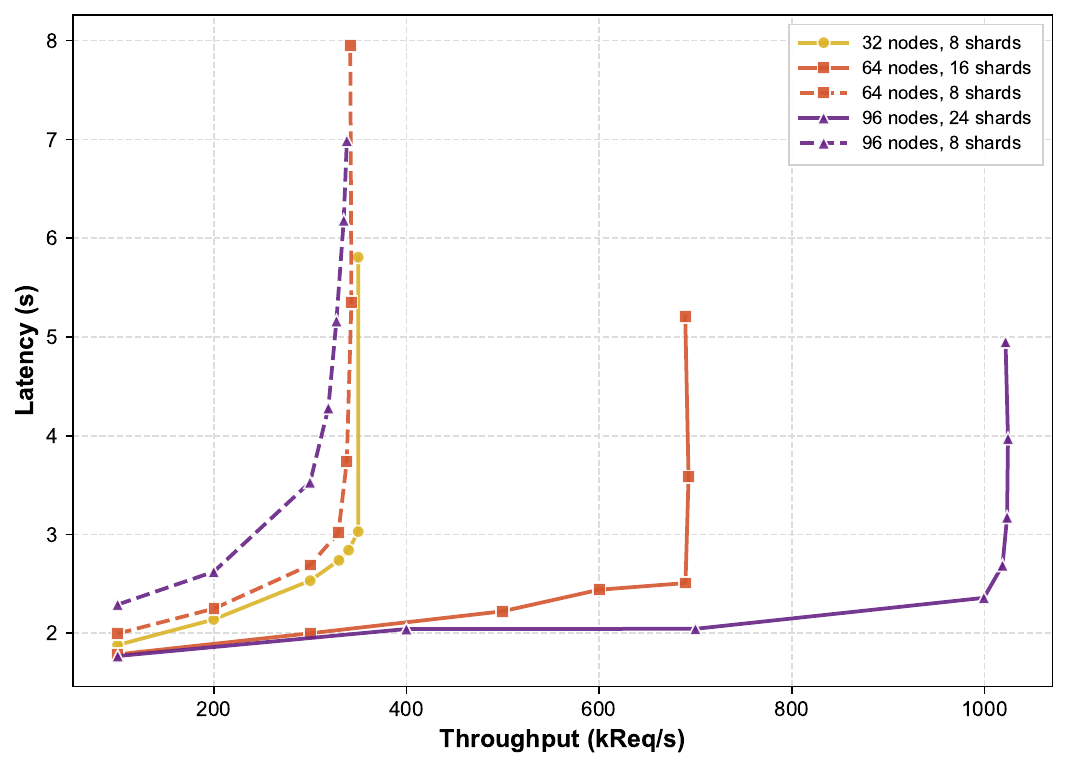}
    \caption{\QuarkShard{} Scalability.}
    \label{fig:scalability}
\end{figure}

\begin{figure}[t]
    \centering
    \begin{subfigure}[b]{0.23\textwidth}
        \centering
        \includegraphics[width=\textwidth]{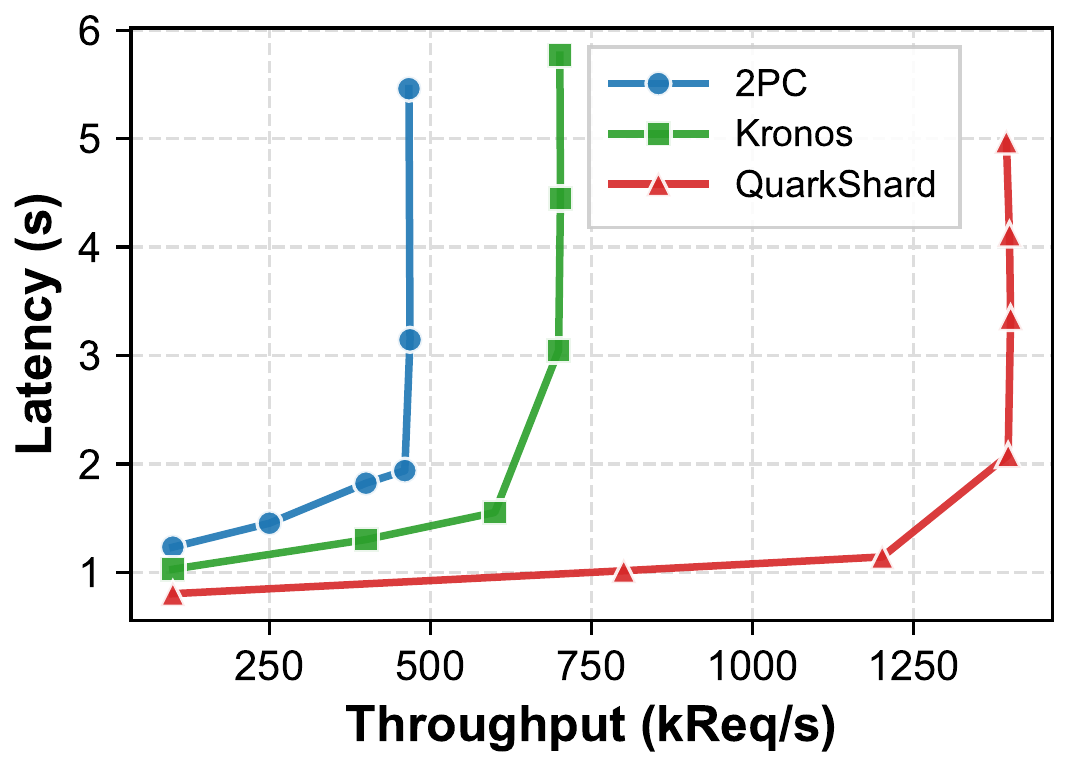}
        \caption{1 Input, 1 Output.}
        \label{fig:lt1i1o}
    \end{subfigure}
    \begin{subfigure}[b]{0.23\textwidth}
        \centering
        \includegraphics[width=\textwidth]{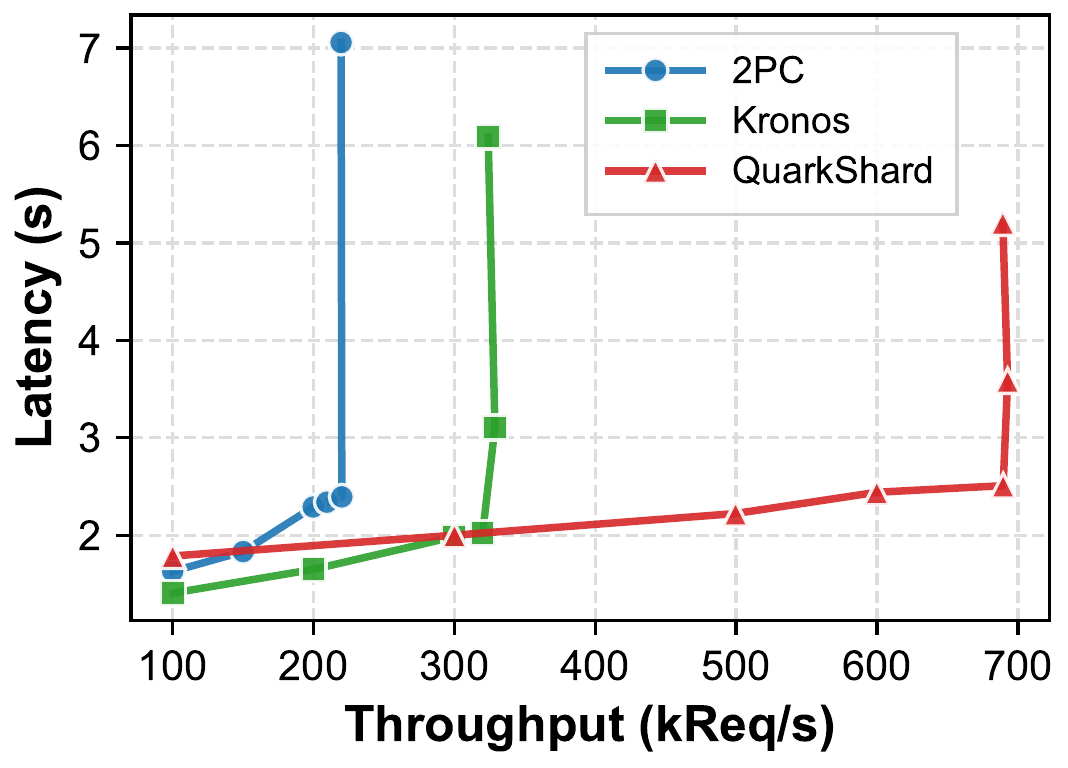}
        \caption{2 Input, 2 Output.}
        \label{fig:lt2i2o}
    \end{subfigure}
    \caption{\QuarkShard{} vs. Baselines.}
    \label{fig:lt}
\end{figure}

We first evaluate the scalability of \QuarkShard{} by measuring its throughput and latency variations under increasing workloads across different numbers of nodes and shards, using requests with two input shards and two output shards ($2$-input-$2$-output) as the benchmark workload (Figure~\ref{fig:scalability}).
Experimental results demonstrate that system throughput grows continuously with the number of shards and nodes, reaching a peak throughput of $1024$\,k requests per second under a $96$-node, $24$-shard deployment.
Concurrently, the end-to-end latency remains low, staying below $4$ seconds under normal workloads even when individual shards are large ($12$ nodes per shard).
This demonstrates that \QuarkShard{} achieves near-linear throughput scalability with respect to shard counts while consistently maintaining low processing latency.

Next, we benchmark the performance of \QuarkShard{}, Kronos, and 2PC under a $64$-node, $16$-shard deployment using two workload profiles: $1$-input-$1$-output requests (Figure~\ref{fig:lt1i1o}) and $2$-input-$2$-output requests (Figure~\ref{fig:lt2i2o}).
Experimental results indicate that across different workloads, \QuarkShard{} achieves $199.6\%\text{--}210.7\%$ of Kronos's maximum throughput and $299.1\%\text{--}314.8\%$ of 2PC's maximum throughput.

Under the single-input request workload, because \QuarkShard{} requires only one round of BFT consensus alongside one round of credit batch importation, its end-to-end request latency is significantly lower than the two-round BFT consensus required by 2PC and Kronos.
Specifically, the end-to-end latency of \QuarkShard{} is only $78.1\%$ of Kronos and $65.2\%$ of 2PC.
Under the multi-input request workload, \QuarkShard{} trades slightly higher latency for substantially higher throughput.
Specifically, \QuarkShard{} incurs an latency approximately $127\%$ of Kronos and $110\%$ of 2PC.

\subsection{Performance under Varying Numbers of Inputs and Outputs}

\begin{figure}[t]
    \centering
    \begin{subfigure}[b]{0.23\textwidth}
        \centering
        \includegraphics[width=\textwidth]{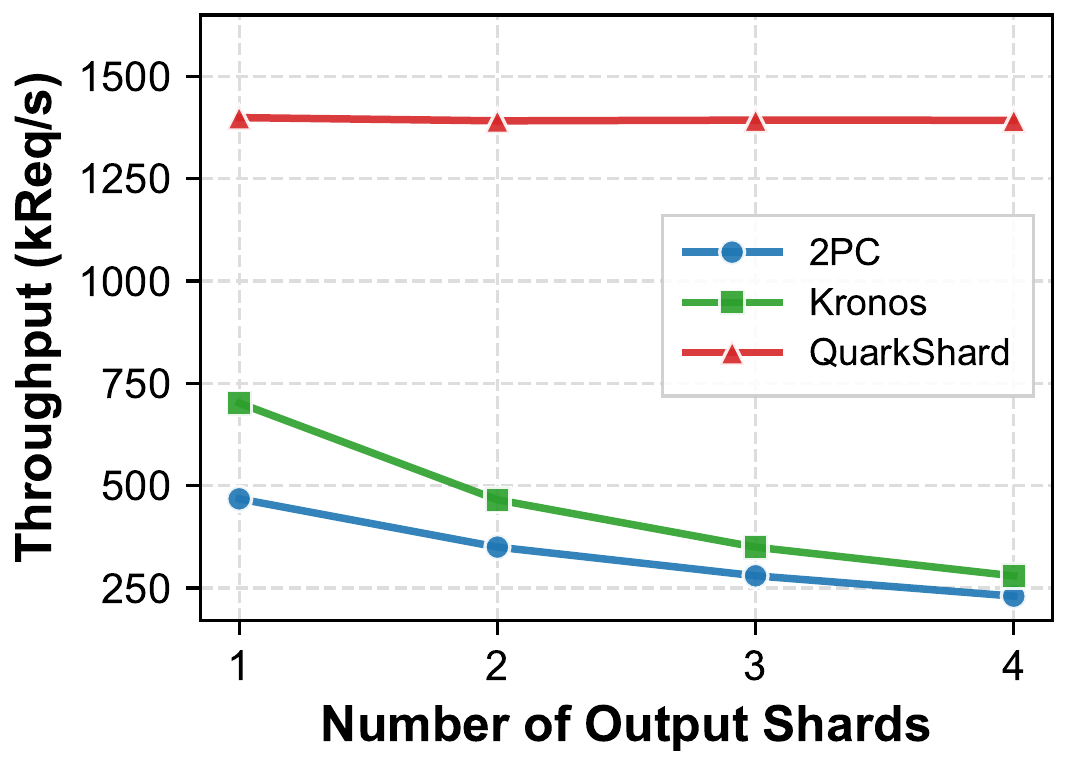}
        \caption{Throughput vs. \# of Outputs.}
        \label{fig:tpsvso}
    \end{subfigure}
    \begin{subfigure}[b]{0.23\textwidth}
        \centering
        \includegraphics[width=\textwidth]{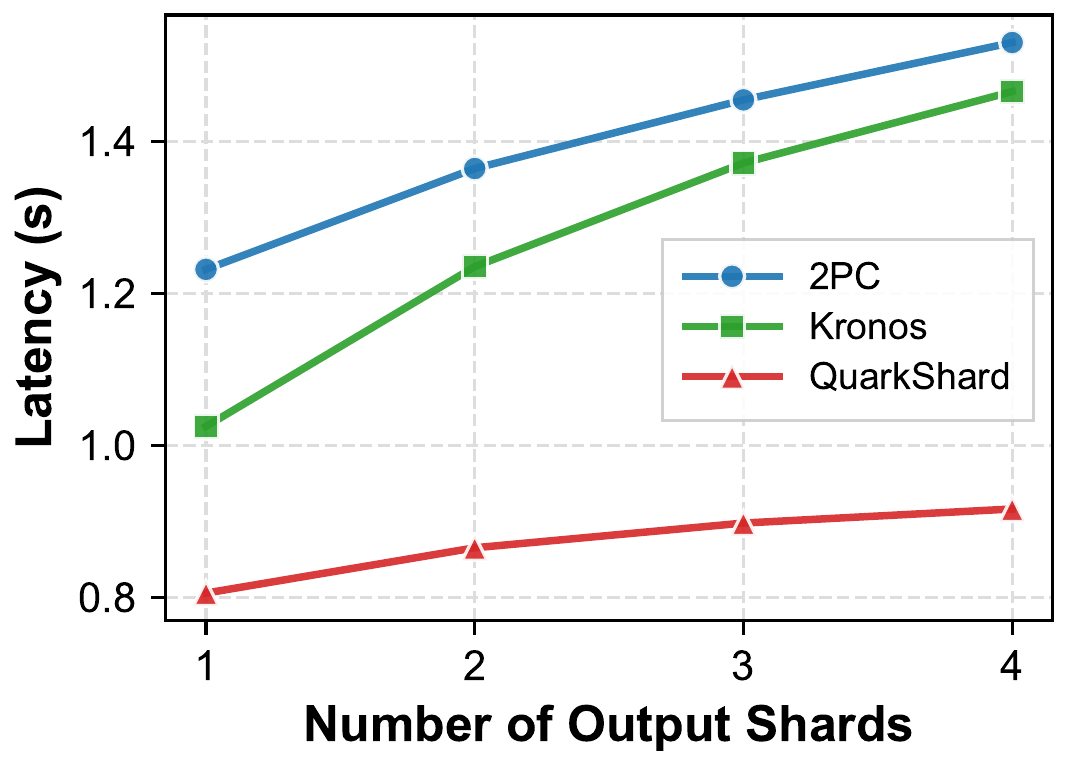}
        \caption{Latency vs. \# of Outputs.}
        \label{fig:latencyvso}
    \end{subfigure}
    \\
    \begin{subfigure}[b]{0.23\textwidth}
        \centering
        \includegraphics[width=\textwidth]{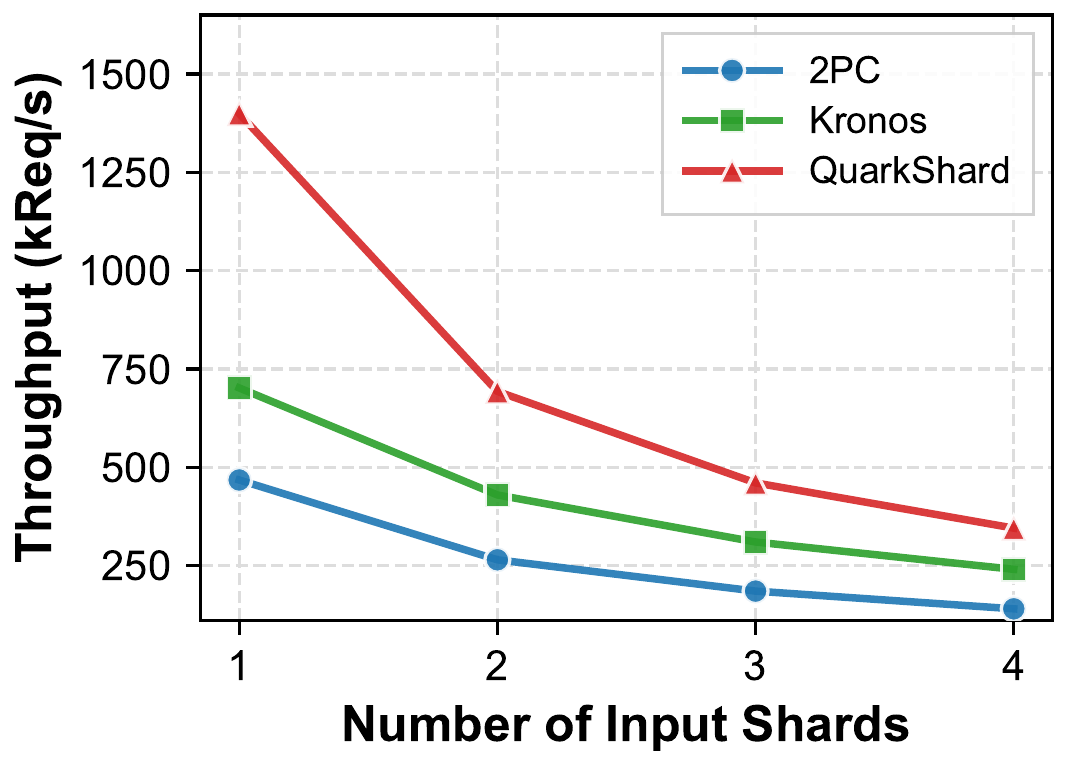}
        \caption{Throughput vs. \# of Inputs.}
        \label{fig:tpsvsi}
    \end{subfigure}
    \begin{subfigure}[b]{0.23\textwidth}
        \centering
        \includegraphics[width=\textwidth]{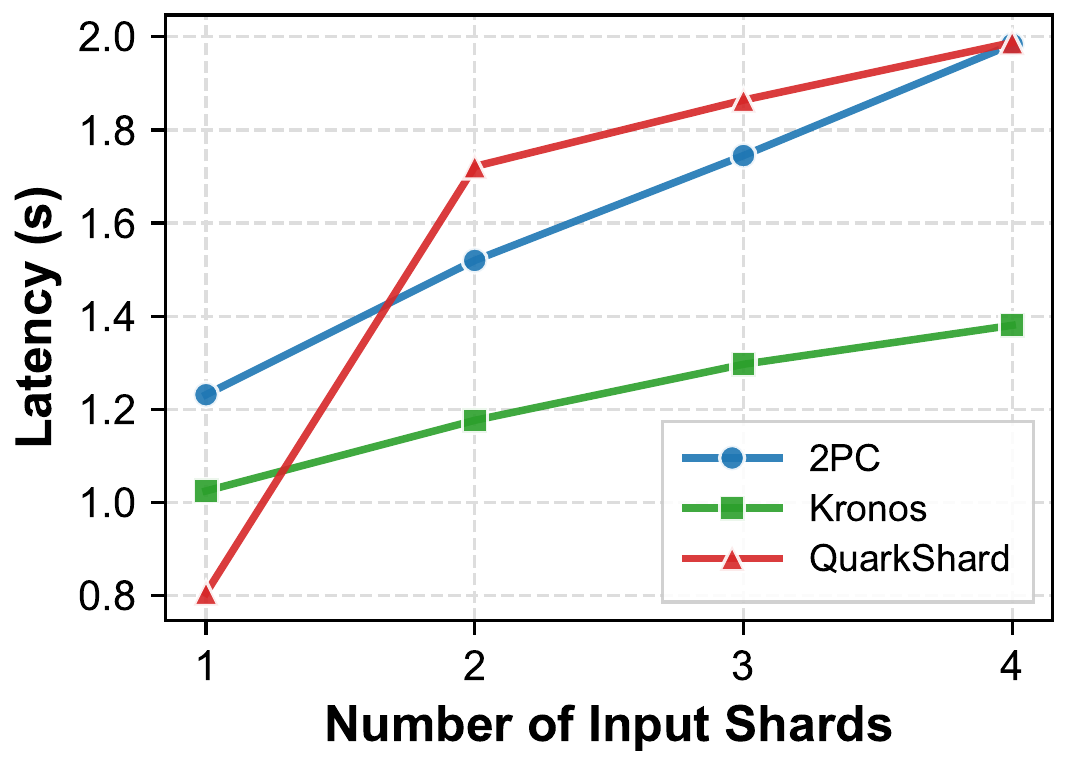}
        \caption{Latency vs. \# of Inputs.}
        \label{fig:latencyvsi}
    \end{subfigure}
    \caption{Performance under Different Workload Profiles.}
    \label{fig:workload}
\end{figure}

To evaluate the impact of different workload characteristics across protocols, we first fix the number of input shards to $1$ and measure how varying the number of output shards affects throughput and latency (Figure~\ref{fig:tpsvso} and \ref{fig:latencyvso}).
We then fix the number of output shards to $1$ and evaluate the system under varying numbers of input shards (Figure~\ref{fig:tpsvsi} and \ref{fig:latencyvsi}).
The experiments were conducted under a deployment of 64 nodes and 16 shards.
Experimental results demonstrate that \QuarkShard{} consistently outperforms 2PC and Kronos in throughput across all tested input and output shard configurations.
Furthermore, except for the transition from $1$ to $2$ input shards, the latency of \QuarkShard{} exhibits lower sensitivity and grows slower with respect to increasing input and output shard counts.

When fixing the number of input shards, variations in the number of output shards have virtually no impact on the throughput of \QuarkShard{}: scaling from $1$ to $4$ output shards results in a throughput degradation of less than $0.5\%$.
In contrast, the throughputs of 2PC and Kronos decrease by $50.9\%$ and $60.2\%$, respectively.
This highlights that in \QuarkShard{}, the computational and communication overhead associated with credit batch importation is negligible compared to transaction execution.
Regarding latency, as the number of output shards scales from $1$ to $4$, the latency of \QuarkShard{} increases by approximately $13.7\%$, compared to $24.2\%$ for 2PC and $43.0\%$ for Kronos.

When fixing the number of output shards, the throughput of \QuarkShard{} is strictly inversely proportional to the number of input shards: the throughput under $1$ input shard is approximately $2.02\times$ that under $2$ input shards and $4.05\times$ that under $4$ input shards.
This is consistent with the \(a\mathcal{B}\) IS-overhead of successful requests under the base compilation algorithm, which grows linearly with the number of input shards.
As for latency, when the number of input shards increases from $2$ to $4$, the latency of \QuarkShard{} increases by approximately $15.5\%$, whereas 2PC and Kronos experience increases of roughly $30.5\%$ and $17.3\%$, respectively.

\subsection{Cross-Shard Communication Overhead}

\begin{figure}[t]
    \centering
    \includegraphics[width=0.47\textwidth]{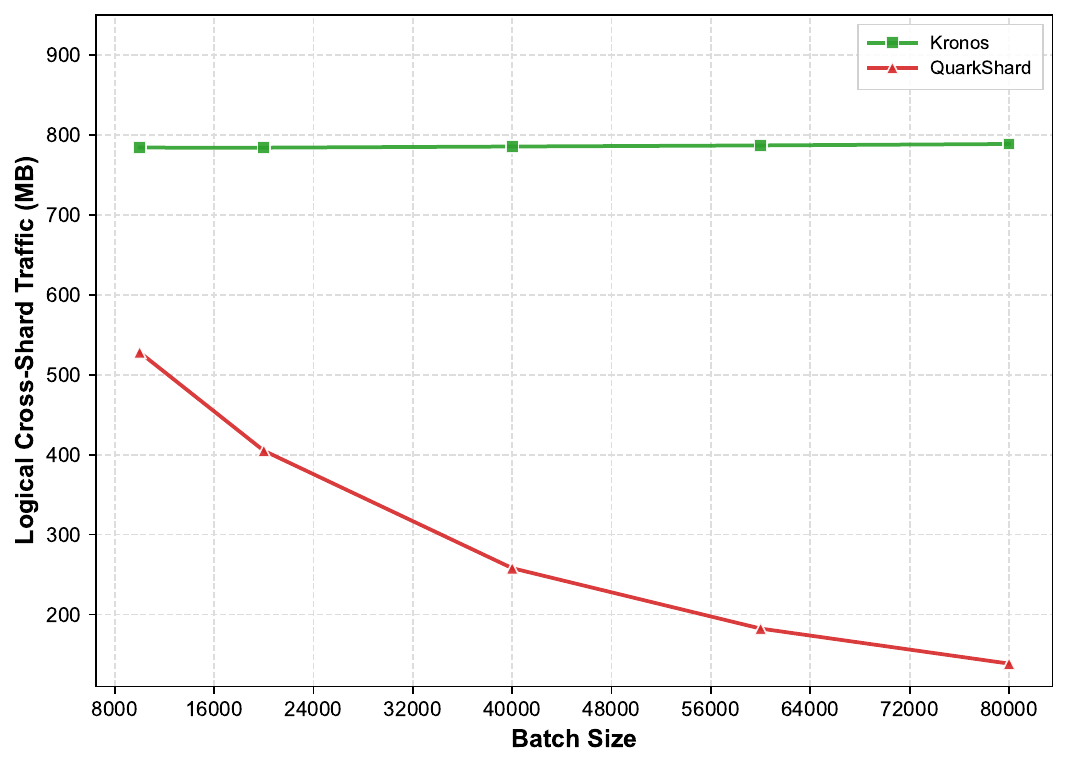}
    \caption{Cross-Shard Communication Overhead.}
    \label{fig:CS-overhead}
\end{figure}

To evaluate the effectiveness of aggregating transaction outputs on a per-account basis in reducing cross-shard communication overhead, we measure the total logical cross-shard communication payload across various batch sizes (i.e., the number of transactions committed per block).
The evaluation is conducted under a deployment of $64$ nodes across $16$ shards, with $1,000$ active accounts per shard.
Experimental results demonstrate that the logical cross-shard communication volume of \QuarkShard{} is consistently lower than that of Kronos. Furthermore, as the block batch size increases, per-account aggregation becomes increasingly effective, leading to a continuous reduction in amortized cross-shard communication payload.
This empirically validates the efficacy of cross-shard transaction outputs aggregating.

\section{Discussion}
\label{sec:discussion}

\subsection{Empty Credit Batch Compression.}
\label{sec:discussion:emptybatch}

In Algorithm~\ref{alg:quark}, we require that for each block and every external shard, a CSRD primitive invocation be executed regardless of whether the corresponding credit batch is empty.
Here, we introduce a straightforward optimization that eliminates the potential overhead of transmitting empty batches while preserving protocol security.

Specifically, each credit batch no longer represents outputs generated from a single block height, but rather spans a contiguous range of block heights.
We denote by \(\DeltaSet{i}{j}{s,h}[pk]\) a non-empty credit batch generated at block height \(h\), signifying that all blocks within the height interval \([s,h)\) in shard \(\Shard{i}\) contained no transaction outputs destined for shard \(\Shard{j}\).
Upon committing a block, replicas omit transmissions to shards with empty batches and instead send credit batches exclusively to shards with non-empty outputs, embedding the starting height \(s\) of the preceding sequence of empty heights into the batch header.
The destination shard can import \(\DeltaSet{i}{j}{s,h}[pk]\) if and only if all heights in \([0,s)\) have already been imported, which atomically effects the importation of all batches within the interval \([s,h]\).

This modification preserves the security of credit batch importation: for every non-empty height, honest replicas transmit exactly one identical credit batch derived from a consensus-committed block.
Consequently, honest nodes in the destination shard receive authenticated credit batches solely for non-empty heights.
Because each credit batch explicitly encodes range metadata, destination nodes maintain monotonic and contiguous importation progress, guaranteeing that every batch is imported exactly once.
Furthermore, since importing empty batches yields no observable state transitions, compressing empty batches into subsequent non-empty batches introduces no additional processing latency.

\subsection{Application on UTXO Model.}

\QuarkShard{} can be readily adapted to the UTXO model.
Specifically, by treating every input and output tuple \(\anglebk{pk,v,\Shard{i}}\) as a UTXO of value \(v\) assigned to account \(pk\) managed by shard \(\Shard{i}\), all preceding design logic and theoretical guarantees hold directly.
Instead of tracking explicit account balances, each replica maintains a local UTXO set, replacing balance additions and deductions with UTXO insertions and deletions.
Cross-shard credit batches remain aggregated on a per-account basis, represented as a single consolidated UTXO for each account, and are delivered to destination shards via the CSRD primitive.
When evaluating transaction validity, replicas verify whether the UTXOs consumed by the transaction exist within the local UTXO set.

\section{Conclusion}

We propose \QuarkShard{}, a novel sharding blockchain protocol that breaks away from the traditional cross-shard request processing paradigm that treats each user request as a monolithic entity.
By compiling and decomposing a request into independent Quark transactions, \QuarkShard{} processes cross-shard requests with reduced overhead while guaranteeing fund atomicity or location atomicity through a carefully designed request compilation algorithm.
This paradigm shift enables \QuarkShard{} to achieve intra-shard BFT execution overhead and cross-shard communication overhead below the theoretical lower bounds imposed by conventional paradigms.
Furthermore, \QuarkShard{} reduces the sharding-specific attack surface for cross-shard front-running and improves resilience to invalid-request flooding.
Comprehensive empirical evaluations demonstrate the low overhead and superior performance of \QuarkShard{} over state-of-the-art protocols.

\cleardoublepage
\appendix

\bibliographystyle{plainurl}
\bibliography{base}

\section{Related Works}

\subsection{Cross-Shard Requests Processing}

A substantial body of prior work has investigated the processing of cross-shard user requests in sharding blockchain systems, proposing various mechanisms to enhance performance~\cite{kokoris2018omniledger,al2017chainspace,zamani2018rapidchain,dang2019towards,wang2019monoxide,hellings2021byshard,huang2022brokerchain,liu2024kronos,liu2026logos,liu2024cherubim,das2020efficient,hong2021pyramid,hong2022scaling,jiang2024sharon,tian2026controlshard,li2023cochain}.

State-of-the-art approaches include Kronos~\cite{liu2024kronos} and Logos~\cite{liu2026logos}.
For a cross-shard user request, Kronos first executes BFT consensus across all input shards to consume the inputs.
The consumed inputs are then transmitted via a trusted cross-shard channel to the output shards, where they are temporarily staged.
Subsequently, output shards execute consensus to transfer all staged funds to the target accounts.
If invalid inputs are encountered, Kronos requires all shards that have already consumed funds to execute consensus on a refund transaction to revert the execution.
In contrast, Logos first broadcasts the user request to all involved shards, whereupon all input shards execute a cross-state determine reliable broadcast (csdRBC) protocol to verify the validity of all inputs.
The resulting validity determinations are then delivered to all involved shards via a trusted cross-shard channel.
If all inputs are valid, all input and output shards execute consensus on the user request to perform the fund deductions and additions.
If any input is invalid, all shards terminate the process directly without requiring further consensus, though the initial broadcast overhead is still incurred.
To optimize cross-shard communication, Logos introduces RPST, which offloads the majority of communication volume to intra-shard networks.

In addition, several other cross-shard request processing schemes have been proposed.
A widely adopted approach is Two-Phase Commit (2PC)~\cite{al2017chainspace,dang2019towards,hellings2021byshard,liu2024cherubim}.
For a cross-shard request, 2PC first executes BFT consensus across all input shards to lock funds, followed by a second consensus round across all input and output shards to actually deduct or credit balances.
Another common approach utilizes relay accounts or user coordinators~\cite{huang2022brokerchain,zamani2018rapidchain}:
to process a request, funds are first transferred via BFT consensus in all input shards to relay accounts or coordinator users in the output shards,
followed by a second consensus round in the output shards to transfer funds from the relay or coordinator accounts to the final destinations.
Monoxide~\cite{wang2019monoxide} and CoChain~\cite{li2023cochain} adopt a relay transaction mechanism analogous to relay accounts: a relay transaction first undergoes BFT consensus in the input shard to transfer funds toward the target shard, after which its output is transferred to the final destination within the target shard.

Nevertheless, these systems adhere to the traditional processing paradigm that treats each user request as a monolithic entity for consensus and execution.
Consequently, they cannot overcome the lower bounds of an IS-overhead of at least \((a+b)\mathcal{B}\) and a CS-overhead of at least \(\Omega(r\lambda)\)~\cite{liu2024kronos,liu2026logos}.

Other schemes introduce specialized sharding architectures to handle cross-shard requests.
Das et al.~\cite{das2020efficient} introduces a dedicated coordinator shard responsible for orchestrating cross-shard requests across different shards.
Tian et al.~\cite{tian2026controlshard} employs specialized shards to process cross-shard requests.
Hong et al.~\cite{hong2021pyramid,hong2022scaling} introduces a hierarchical sharding structure in which bridge shards cover multiple lower-level shards and coordinate their cross-shard requests.
Xu et al.~\cite{xu2023two} similarly employ a two-tiered sharding architecture, in which a set of small execution shards is overseen by a control shard to process user requests.
Jiang et al.~\cite{jiang2024sharon} proposed a method that dynamically merges pairs of shards and periodically rotates the pairings, permitting cross-shard requests to execute only when the involved shards are actively merged.

\subsection{MEV in sharding Blockchains}

Extracting Maximal Extractable Value (MEV) through transaction order manipulation has emerged as a major security challenge in blockchain systems~\cite{daian2020flash,qin2022quantifying}.

Order-fairness BFT protocols, such as Aequitas~\cite{kelkar2020order}, investigate mechanisms to prevent malicious leaders or proposers from arbitrarily reordering transactions observed by honest nodes.
However, these protocols primarily address local transaction ordering within a single consensus instance.
In sharding blockchains, achieving fair ordering for cross-shard requests typically necessitates specialized protocol designs. Prophet~\cite{hong2023prophet} adopts a Byzantine fault-tolerant deterministic ordering strategy to prevent conflicts among cross-shard transactions.

Furthermore, sharding blockchains introduce unique fairness challenges.
In conventional cross-shard protocols, a cross-shard request must reach consensus within its input shards before undergoing a separate consensus process within its output shards.
This sequential workflow inherently exposes cross-shard requests to higher latency compared to intra-shard transactions, thereby expanding the attack surface for front-running exploits.
Haechi~\cite{zhang2023front} identifies front-running attacks stemming from the execution window gap between input and output shard consensus instances in traditional cross-shard workflows, and proposes finalization fairness to align the processing and execution order across shards.

\section{CSRD Construction}
\label{app:csrd}

\subsection{Building Blocks}

We use the following standard building blocks. Throughout this section, \(q=n-f\) and \(k=f+1\).

\paragraph{Threshold signatures.}
A \((q,n)\) threshold signature scheme provides:
\begin{itemize}[leftmargin=*]
  \item \(\TSigSignShare{\ell}{x}\to s_\ell\), invoked by replica \(p_\ell\) to generate its signature share on \(x\);
  \item \(\TSigVerifyShare{\ell}{x}{s_\ell}\), which verifies a signature share;
  \item \(\TSigCombine{x}{\{s_\ell\}}\to \sigma\), which aggregates \(q\) valid signature shares from distinct replicas into a threshold signature \(\sigma\);
  \item \(\TSigVerify{x}{\sigma}\), which verifies a threshold signature \(\sigma\).
\end{itemize}
Combining \(q\) valid signature shares from distinct replicas yields an accepted signature.
Conversely, if fewer than \(q\) replicas contribute signature shares, generating an accepted signature is computationally infeasible.

\paragraph{Erasure coding.}
A \((k,n)\) erasure code provides:
\begin{itemize}[leftmargin=*]
  \item \(\ECEnc(M)\to\{e_1,\ldots,e_n\}\), which encodes a message \(M\) into \(n\) fragments;
  \item \(\ECDec(\{\anglebk{\ell,e_\ell}\}_{\ell\in I})\to M\), which reconstructs the original message \(M\) using a set of fragments indexed by \(I\), where \(I\) is an index set of size at least \(k\).
\end{itemize}
Any \(k\) distinct valid fragments guarantee the exact reconstruction of the original message. The encoding is deterministic, and each fragment has a size of \(O(|M|/k)\).

\paragraph{Commitments.}
A commitment scheme provides:
\begin{itemize}[leftmargin=*]
    \item \(\COMCommitOpen(\{e_1,\ldots,e_n\},\ell)\to(c, \mu_\ell)\), which returns a commitment \(c\) and an opening \(\mu_\ell\);
    \item \(\COMVerify(c,\ell,e_\ell,\mu_\ell)\), which verifies that \(e_\ell\) occupies position \(\ell\) in the committed vector.
\end{itemize}
We require that the commitment is deterministic.
And we require correctness and position binding: after \(c\) is fixed, no adversary can open the same position to two different values.

We keep this interface generic, but all commitment instances in this paper are instantiated via a Merkle tree~\cite{merkle1987digital} in which \(e_\ell\) serves as the \(\ell\)-th leaf node.
Under this construction, the commitment is the Merkle root, and an opening is a Merkle proof for the corresponding leaf.
Thus, \(c\) has a size of \(O(\lambda)\), and an opening has a size of \(O(\lambda\log n)\).

\paragraph{Intra-shard certification.}
The Logos-based construction additionally employs an intra-shard certification primitive. This primitive provides:
\begin{itemize}[leftmargin=*]
  \item \(\mathsf{Certify}_{i}(x)\to\rho\), invoked by replica \(p_\ell\) to produce a certificate for \(x\);
  \item \(\mathsf{CertVerify}(x,\rho)\), which verifies whether a certificate is valid.
\end{itemize}

An accepted certificate can be obtained by all replicas if and only if at least \(q\) replicas invoke \(\mathsf{Certify}_{i}(x)\). In general, this certification primitive does not introduce additional consensus rounds; instead, it can be seamlessly embedded into intra-shard BFT consensus. A typical instantiation involves embedding the hash of \(x\) into the header of a BFT proposal block, using the block commit certificate together with the Merkle proof of inclusion for \(x\)'s hash in the block header as the validity certificate \(\rho\).

\subsection{A Basic CSRD Instance}
\label{app:csrd:instance}

In Algorithm~\ref{alg:csrd}, each replica in the input shard broadcasts its threshold signature share along with its designated erasure-coded fragment corresponding to its replica index to all replicas in the target shard.
Simultaneously, cryptographic commitments are used to protect the integrity of the fragments.
Upon receiving \(n-f\) messages containing valid signature shares, a replica in the target shard can aggregate these shares to form a certificate.
Among these messages, at least \(n-2f \geq f+1\) originate from honest replicas; thus, the original data can be reconstructed using the accompanying erasure-coded fragments.

\subsection{CSRD Instance based on Logos}
\label{app:csrd:logos}

Algorithm~\ref{alg:csrd-logos} adopts the one-to-one cross-shard transmission pattern proposed by Logos~\cite{liu2026logos}.
Specifically, a source replica with index \(\ell\) transmits fragment \(\ell\) exclusively to the target replica sharing the same index \(\ell\), which subsequently propagates the fragment within the destination shard.
Unlike Algorithm~\ref{alg:csrd}, where target shard replicas aggregate the certificate, source replicas in Algorithm~\ref{alg:csrd-logos} first generate the certificate intra-shard before transmitting it to the destination.

\begin{algorithm}[t]
\footnotesize
\DontPrintSemicolon
\AlgTwoERole{\textbf{as the} \(\ell\)-th \textbf{replica} \(p_\ell\in\Shard{i}\)}
\AlgUpon{\(\csrd{i}{j}.\Send{\tau,M}\) is called\label{alg:csrd:send}}{
    \(s_\ell\gets\TSigSignShare{\ell}{\anglebk{i, j, \tau}}\).\;
    \(\bigbk{e_1,\ldots,e_n}\gets\ECEnc(M)\).\;
    \((c,\mu_\ell)\gets\COMCommitOpen(\bigbk{e_1,\ldots,e_n},\ell)\).\;
    Send \(\text{msg}_\ell=\anglebk{\tau,c,\ell,e_\ell,\mu_\ell,s_\ell}\) to every replica of \(\Shard{j}\).\;
}
\AlgTwoESep
\AlgTwoERole{\textbf{as a replica in} \(\Shard{j}\)}
\AlgUpon{receiving \(n-f\) \(\text{msg}_\ell\) from distinct source replicas \(p_\ell\), each satisfying \(\TSigVerifyShare{\ell}{\anglebk{i,j,\tau}}{s_\ell}\)}{
    \(\pi\gets\TSigCombine{\anglebk{i,j,\tau}}{\{s_\ell\}}\).\;
    Find \(I\) of \(f+1\) distinct indices whose messages have the same \(c\) and satisfy \(\COMVerify(c,\ell,e_\ell,\mu_\ell)\).\;
    \(M\gets\ECDec(\{(\ell,e_\ell)\}_{\ell\in I})\).\;
    Deliver \(\anglebk{M,\pi}\).\;
}
\AlgTwoESep
\AlgTwoERole{\textbf{as a replica in} \(\Shard{i}\) \textbf{or} \(\Shard{j}\)}
\AlgUpon{\(\csrd{i}{j}.\Verify{\tau,\pi}\) is called\label{alg:csrd:verify}}{
    \Return \(\TSigVerify{\anglebk{i,j,\tau}}{\pi}\).\;
}
\caption{A basic CSRD instance from \(\Shard{i}\) to \(\Shard{j}\) with tag \(\tau\).}
\label{alg:csrd}
\end{algorithm}

\begin{algorithm}[t]
\footnotesize
\DontPrintSemicolon
\AlgTwoERole{\textbf{as the} \(\ell\)-th \textbf{replica} \(p_\ell\in\Shard{i}\)}
\AlgUpon{\(\csrd{i}{j}.\Send{\tau,M}\) is called}{
    \(\bigbk{e_1,\ldots,e_n}\gets\ECEnc(M)\).\;
    \((c,\mu_\ell)\gets\COMCommitOpen(\bigbk{e_1,\ldots,e_n},\ell)\).\;
    \(\rho\gets\mathsf{Certify}_{i}(\anglebk{i,j,\tau,c})\).\;
    Send \(\text{msg}_\ell=\anglebk{\tau,c,\ell,e_\ell,\mu_\ell,\rho}\) to the \(\ell\)-th replica of \(\Shard{j}\).\;
}
\AlgTwoESep
\AlgTwoERole{\textbf{as a replica in} \(\Shard{j}\)}
\AlgUpon{receiving \(\text{msg}_\ell\) from the \(\ell\)-th replica of \(\Shard{i}\), satisfying \(\mathsf{CertVerify}_{i}(\anglebk{i,j,\tau,c},\rho)\) and \(\COMVerify(c,\ell,e_\ell,\mu_\ell)\)}{
    Broadcast \(\text{msg}_\ell\) to all replicas in \(\Shard{j}\).\;
}
\AlgUpon{receiving valid \(\text{msg}_\ell\) for a set \(I\) of \(f+1\) distinct indices, all with the same \(c\)}{
    \(M\gets\ECDec(\{(\ell,e_\ell)\}_{\ell\in I})\).\;
    \(\pi\gets\anglebk{c,\rho}\) for any valid \(\rho\) among these messages.\;
    Deliver \(\anglebk{M,\pi}\).\;
}
\AlgTwoESep
\AlgTwoERole{\textbf{as a replica in} \(\Shard{i}\) \textbf{or} \(\Shard{j}\)}
\AlgUpon{\(\csrd{i}{j}.\Verify{\tau,\pi}\) is called}{
    Parse \(\pi=\anglebk{c,\rho}\).\;
    \Return \(\mathsf{CertVerify}(\anglebk{i,j,\tau,c},\rho)\).\;
}
\caption{A Logos-based CSRD instance from \(\Shard{i}\) to \(\Shard{j}\) with tag \(\tau\).}
\label{alg:csrd-logos}
\end{algorithm}

For ease of proof, the algorithm specifies that certificate generation occurs after replicas initiate transmission.
However, in practical deployments, the certificate may be generated alongside intra-shard BFT consensus prior to transmission.
This does not compromise security, as our model requires that any honest replica observing a valid certificate will invariably invoke the transmission procedure.

\subsection{Analysis}

\begin{lemma}[CSRD security]
\label{lem:csrd-security}
Under authenticated reliable channels, \(n\geq 3f+1\), and the properties of the building blocks above, Algorithms~\ref{alg:csrd} and \ref{alg:csrd-logos} satisfy CSRD correctness and validity.
\end{lemma}

\begin{proof}
For Algorithm~\ref{alg:csrd}, suppose all honest source replicas invoke \(\Send{\tau,M}\) with the same \(M\).
Deterministic encoding and commitment make them send fragments under the same commitment \(c^*\).
Every honest destination replica eventually receives messages from all \(n-f\) honest source replicas and can therefore combine their signature shares.
Any \(n-f\) received messages contain at least \(n-2f\geq f+1\) honest messages, so a qualifying set of \(f+1\) messages exists.
Moreover, every set of \(f+1\) messages from distinct source replicas that shares a commitment contains an honest message.
Its commitment must consequently be \(c^*\); position binding makes every accepted fragment in that set the corresponding fragment of the codeword for \(M\).
Erasure-code correctness therefore reconstructs the unique \(M\).
This proves correctness.

For validity, an accepted threshold signature contains \(n-f\) distinct shares, at least \(n-2f\) of which came from honest source replicas.
Since an honest replica signs \(\anglebk{i,j,\tau}\) only while processing \(\csrd{i}{j}.\Send{\tau,M}\), at least one honest source replica invoked \(\Send{\tau,M}\) for this instance.
Notice that this argument does not require the signature to cover \(M\), a digest, or \(c\).

For Algorithm~\ref{alg:csrd-logos}, deterministic encoding again gives all honest source replicas the same \(c^*\), and certification correctness supplies an accepted certificate for it.
Among the \(n-f\) honest sources, at most \(f\) have a Byzantine replica at the corresponding destination index.
Hence at least \(n-2f\geq f+1\) correctly opened fragments reach honest paired destinations and are broadcast to every honest destination replica.
Position binding and erasure-code correctness make every such replica reconstruct the same \(M\).

For validity, an accepted certification proof has \(n-f\) participants, at least \(n-2f\) of which are honest; each honest participant certifies only while processing \(\Send{\tau,M}\).
Thus an accepted CSRD certificate implies that an honest source replica invoked \(\Send{\tau,M}\).
Certification agreement additionally prevents the same instance from certifying two different commitments.
\end{proof}

\begin{lemma}[CSRD cross-shard overhead]
\label{lem:csrd-overhead}
Ignoring lower-order authentication metadata when the encoded message dominates it, Algorithms~\ref{alg:csrd} and \ref{alg:csrd-logos} transmit
\[
    O(n^2|M|/(f+1))
    \quad\text{and}\quad
    O(n|M|/(f+1))
\]
bytes across shards, respectively. Under the standard configuration \(n=3f+1\), these become \(O(n|M|)\) and \(O(|M|)\).
\end{lemma}

\begin{proof}
Each encoded fragment has size \(O(|M|/(f+1))\).
In Algorithm~\ref{alg:csrd}, every cross-shard message additionally contains the tag \(\tau\), source index \(\ell\), commitment \(c\), opening \(\mu_\ell\), and signature share \(s_\ell\).
Thus, the size of the metadata in each message is
\[
    \eta_{\mathsf{B}}(n)
    =
    |\tau|+\lceil\log n\rceil+|c|+|\mu_\ell|+|s_\ell|,
\]
and sending one such message from each of \(n\) source replicas to each of \(n\) destination replicas costs
\[
    O\left(\frac{n^2|M|}{f+1}+n^2\eta_{\mathsf{B}}(n)\right).
\]

Let \(p\) denote the number of entries covered by the export root used for certification, taking \(p=1\) when the certified value is committed directly in the block header.
In Algorithm~\ref{alg:csrd-logos}, a cross-shard message instead contains \(\tau\), \(\ell\), \(c\), \(\mu_\ell\), and the intra-shard certificate \(\rho\).
Hence, the size of the metadata in each message is
\[
    \eta_{\mathsf{L}}(n,p)
    =
    |\tau|+\lceil\log n\rceil+|c|+|\mu_\ell|+|\rho|.
\]
Only one message is transmitted along each of the \(n\) matched source--destination pairs, yielding a total cost of
\[
    O\left(\frac{n|M|}{f+1}+n\eta_{\mathsf{L}}(n,p)\right).
\]
The certification procedure and the subsequent destination-shard rebroadcast are intra-shard operations, so neither contributes to CS-overhead.

With the default instantiation, \(|\tau|,|c|,|s_\ell|=O(\lambda)\), while \(|\mu_\ell|=O(\lambda\log n)\).
The certificate \(\rho\) consists of a constant-size threshold-signature commit certificate and, when needed, a Merkle inclusion proof with respect to this export root.
Consequently, \(|\rho|=O(\lambda\log(p+1))\).
The tag and replica index are absorbed by these terms, yielding
\(\eta_{\mathsf{B}}(n)=O(\lambda\log n)\) and
\(\eta_{\mathsf{L}}(n,p)=O(\lambda(\log n+\log(p+1)))\).
Substitution gives
\[
    O\left(
        \frac{n^2|M|}{f+1}+n^2\lambda\log n
    \right)
\]
and
\[
    O\left(
        \frac{n|M|}{f+1}
        +n\lambda(\log n+\log(p+1))
    \right).
\]
For \(n=3f+1\), these become \(O(n|M|+n^2\lambda\log n)\) and \(O(|M|+n\lambda(\log n+\log(p+1)))\).

Finally, if \(|M|=\Omega(n\lambda(\log n+\log(p+1)))\), the message terms dominate the metadata terms, leaving \(O(n|M|)\) and \(O(|M|)\), respectively.
\end{proof}

\end{document}